\documentclass[11pt]{article}

\usepackage{hyperref}
\usepackage{latexsym}
\usepackage{amsmath,amssymb,amsthm}

\usepackage{setspace}
 \newcommand{\beqn}{\begin{eqnarray}}
 \newcommand{\eeqn}{\end{eqnarray}}
 \newcommand{\be}{\begin{equation}}
 \newcommand{\ee}{\end{equation}}
 \newcommand{\ba}{\begin{array}}
 \newcommand{\ea}{\end{array}}

 \newcommand{\pa}{\partial}

 \newcommand{\ds}{\displaystyle}
 \newcommand{\la}{\label}
 \newcommand{\rIm}{{\rm Im\5}}
 \newcommand{\rRe}{{\rm Re\5}}

\newcommand{\ov}{\overline}
\newcommand{\w}{\rm {w}}
\newcommand{\ti}{\tilde}
\newcommand{\cE}{{\cal E}}
\newcommand{\cH}{{\cal H}}
\newcommand{\ve}{\varepsilon}
\newcommand{\De}{\Delta}

\newcommand{\al}{\alpha}
\newcommand{\ga}{\gamma}

\newcommand{\si}{\sigma}
\newcommand{\om}{\omega}
\newcommand{\Om}{\Omega}
\newcommand{\na}{\nabla}
\newcommand{\lam}{\lambda}
\newcommand{\Lam}{\Lambda}
\newcommand{\5}{{\hspace{0.5mm}}}
\newcommand{\R}{\mathbb{R}}

\newcommand{\C}{\mathbb{C}}

\newtheorem{theorem}{Theorem}[section]

\newtheorem{defin}[theorem]{Definition}
\newtheorem{lemma}[theorem]{Lemma}
\newtheorem{remark}[theorem]{Remark}
\newtheorem{cor}[theorem]{Corollary}
\newtheorem{pro}[theorem]{Proposition}

\begin{document}
\begin{center}
{\Large On 
soliton asymptotics for 2D Maxwell--Lorentz 
 \\
 equations with rotating particle}
 \bigskip \bigskip

E.A. Kopylova\footnote{ 
 Supported partly by Austrian Science Fund (FWF) P34177
 }
 \medskip
 \\
{\it
Faculty of Mathematics of   Vienna University
    }
  \\
  elena.kopylova@univie.ac.at
\begin{abstract}
We consider 2D Maxwell--Lorentz equations with extended charged rotating particle. 
The system admits solitons which are solutions corresponding to a particle moving with a constant velocity and rotating with a constant angular velocity. 
Our main result is asymptotic stability of  the solitons.
\end{abstract}

\end{center}
\setcounter{equation}{0}
\section{Introduction} 
 2D Maxwell--Lorentz equations with the rotating particle  read 
\be\la{mls2}
\left\{\ba{rcl}
\dot E(x,t)&=&J\na B(x,t)-[\dot q(t)-\om(t) J(x-q(t))] \rho(x-q(t))\\
\dot B(x,t)&=&-\na \cdot(JE(x,t)),
\qquad
\na\cdot E(x,t)=\rho (x-q(t))\\
m\ddot q(t)&=&\langle E(x,t)+B(x,t) [J\dot q(t) +\om(t) (x-q(t)) ],\rho(x\!-\!q(t))\rangle\\
I\dot \om(t)&=&\langle (x\!-\!q(t))\cdot \big[J E(x,t)\!-B(x,t) \dot q(t)],\rho(x\!-\!q(t))\rangle
\ea\right|,~~J=\begin{pmatrix} 0& 1\\ -1& 0\end{pmatrix}. 
\ee
The system describes a motion of an extended
charged particle, centered at $q(t)$, in the Maxwell field $(E(x, t), B(x, t))\in \R^2\oplus \R$.  
Here  the brackets $\langle,\rangle$ denote the inner product in the  Hilbert space $L^2:=L^2(\R^2)\otimes\R^2$,
$m >0$ is the mass of the particle and $I>0 $ is its moment of inertia. Note that the system is an analogue of 3D Maxwell--Lorentz system \cite[(2.39)--2.41)]{S2004}.

We assume that the charge density $\rho(x)$ is real-valued, compactly supported and  spherically symmetric, i.e.,
\be\la{rosym}
\rho\in C_0^\infty(\R^2),\quad \rho(x)=0 ~~{\rm for}~~|x|\ge R_{\rho},\quad\rho(x)=\rho_{rad}(|x|),\quad\rho(x)\not\equiv 0.
\ee
Moreover, we assume that 
\be\la{zero1}
\int x^{\al}\rho(x) dx=0,\qquad |\al|\le 4. 
\ee
In particular, the total charge $\int\rho(x) dx$ equals zero (neutrality of the particle). 
Equivalently, 
\be\la{zero2}
\hat\rho^{(\al)}(0)=0,\qquad |\al|\le 4.
\ee
Note that \eqref{zero1} holds if one of $\alpha_j$, $j=1,2$,  or both  of them is odd.
We rewrite the equation  \eqref{mls2} in terms of the Maxwell potentials $A(x,t)$ and $B(x,t)$ (the system  \eqref{mls3}).
The system admits  soliton solutions $(A_{v,\w}, B_{v,\w})$  corresponding to the particle which  move with constant  velocities $v\in\R^2$, $0\le |v|<1$ 
and rotate with constant angular velocities $\w\in\R$. 
Asymptotic stability of the soliton with velocity $v_0$ and angular velocity $\w_0$
means that  for solutions with  initial data  sufficiently close to  this  soliton the following properties are satisfied:
\\
i) $\dot q(t)\to v_\pm\approx v_0$ and $\om(t)\to \om_\pm\approx \w_0$ as $t\to\pm\infty$;
\\
ii) for the  field part of the solutions asymptotics (\ref{AP-as})  hold, where
the remainder $r_{\pm}(x,t)\to 0$ as $t\to\pm\infty$ in the global energy norm. 
\\
In this paper we prove the asymptotic stability of solotons under a certain spectral condition ( condition \eqref{M-condition} below), 
which is  an analogue of the Wiener condition for the system \eqref{mls2}. In particular, the condition is satisfied
for  a sufficiently large $I$.
\\
Our results extend the result of \cite{K2025}, where asymptotic stability of a soliton with $\w=0$  was established.
\setcounter{equation}{0}
\section{Hamiltonian structure and well-posedness}
In the Maxwell potentials $A(x,t)=(A_1(x,t),A_2(x,t))$, $\Phi(x,t)$, we have
\be\la{BAE}
B(x,t)=\na\cdot (JA)=\na_1A_2(x,t)-\na_2 A_1(x,t),\qquad E(x,t)=-\dot A(x,t)-\na \Phi(x,t).
\ee
We choose the Coulomb gauge $\na\cdot A(x,t)=0$. Then the  second line  of  (\ref{mls2}) implies
$$
 -\De\Phi(x,t)=\rho(x-q(t)).
$$
Hence,
\be\la{2mA3}
\Phi(x,t)=\Phi_0(x-q(t)),\qquad \Phi_0(y)=-\frac 1{2\pi}\int \log |x-y| \rho(y)dy.
\ee
We set 
\be\la{pM}
p=m\dot q+\langle A(x),\rho(x-q)\rangle,\quad M=I\om-\langle A(x)\cdot J(x-q),\rho(x-q)\rangle\big).
\ee
By \eqref{mls2} and \eqref{BAE},
\beqn\nonumber
\dot p&=&m\ddot q+\langle  \dot A(x),\rho(x-q)\rangle-\langle A(x),(\dot q\cdot\na)\rho(y)\rangle\\
\nonumber
&=&-\langle  A(x)\cdot \dot q,\na\rho(x-q)\rangle+\om \langle \na\cdot JA(x),  (x-q)\rho(x-q)\rangle,\\
\nonumber
\dot M&=&I\dot \om -\langle \dot A(y+q)\cdot Jy,\rho(y)\rangle-\langle (\na \cdot \dot q)A(y+q)\cdot Jy),\rho(y)\rangle
=-\langle A\cdot \dot q, (\na\cdot J)\rho\rangle=0
\eeqn
Thus, we  can rewrite system  (\ref{mls2}) as
\be\la{mls3}
\!\!\left\{\ba{llll}
\dot A(x,t)=\Pi(x,t)
\\
\dot \Pi(x,t)=\Delta A(x,t)-\om J\varrho(x-q(t))+{\cal P}[\dot q\rho(x-q(t))]\\
m\dot q(t)=p(t)-\langle A(x,t),\rho(x-q(t))\rangle
\\
\dot p(t)=-\langle  A(x,t)\cdot \dot q(t),\nabla\rho(x-q(t))\rangle+\om(t) \langle \na\cdot JA(x,t),  \varrho(x-q(t))\rangle 
\ea\!\right|,~  \varrho(x):=x\rho(x).
\ee
Here
$$
\om(t)=\frac{1}{I}\big(M+\langle A(x,t),J\varrho(x-q(t))\rangle\big).
$$
Denote by ${\cal P}$ the projection onto the space of solenoidal (divergence-free) vector fields,
which in Fourier space reads:
\be\la{Pi-e}
\widehat{{\cal P}a}(k)=\hat a(k)-\frac{\hat a(k)\cdot k}{k^2}k=\frac{\hat a(k)\cdot Jk}{k^2}Jk.
\ee
Let us introduce a phase space for the system \eqref{mls3}. Denote by 
$\dot H^1$  the closure of $C_0^\infty (\R^2)\otimes \R^2$ with the respect to the norm 
$\Vert \nabla A\Vert_{\dot H^1}=\Vert \nabla A\Vert_{L^2}$, where $L^2:=L^2(\R^3)\otimes\R^2$.
Let ${\bf L}^2$, $\dot {\bf H}^1$ be the subspaces formed by solenoidal vector fields.
Define the phase space
\be\la{psp}
{\cal E}={\cal F}\oplus\R^2\oplus\R^2, \quad  {\cal F}= \dot {\bf H}^1\oplus {\bf L}^2.
\ee
System (\ref{mls3}) is  a Hamiltonian system with the Hamiltonian functional
\be\la{Hp}
H(Y)=\frac 12\int[|\Pi(x)|^2+|\na A(x)|^2]dx+\frac 1{2}m\dot q^2(t)+\frac 1{2} I\om^2(t),~~
Y=(A,\Pi, q, p)\in {\cal E}
\ee
where 
\be\la{dqp}
\frac{m\dot q^2}2=\frac1{2m} (p-\langle A,\rho(x-q)\rangle)^2,
\qquad
\frac{I\om^2}2=\frac1{2I}(M+\langle A,J\varrho(x-q)\rangle)^2.
\ee
The Hamiltonian (\ref{Hp}) is well defined and differentiable on the Hilbert
phase space (\ref{psp}).
The system (\ref{mls3}) is equivalent to the canonical Hamiltonian system
$$
\dot A=D_\Pi H,\,\,\,\dot \Pi=-D_A H;\quad\dot q=D_{p} H,\,\,\,\dot p=-D_q H,
$$
or equivalently, 
\be\la{canH}
\dot Y={\bf J}{\cal D}\cH(Y),\qquad 
{\bf J}:=\left(
\ba{ccccc}
0 & I_2 & 0 & 0 \\
-I_2 & 0 & 0 & 0\\
0 & 0 & 0 & I_2 \\
0 & 0 & -I_2 & 0 
\ea
\right).
\ee
where $I_2$ denotes the identity $2\times 2$ matrix.
\begin{pro} 
Let \eqref{rosym} holds, and let $Y_0=(A_0, \Pi_0, q_0, p_0)\in {\cal E}$. Then\\
(i) there exists a unique solution $Y(t)\in C(\R, {\cal E})$ to the Cauchy problem for \eqref{mls3};\\
(ii) the energy conserves: $H(Y(t))=H(Y_0)$ for  $t\in R$;\\
(iii) the estimate holds,
\be\la{apr}
|\dot q|\le \ov v,\quad t\in\R.
\ee
\end{pro}
The proof is similar to  \cite[Lemma A.2]{KS2000}.
\setcounter{equation}{0}
\section{Solitons}
Solitons of the system (\ref{mls3}) are solutions of the form
\be\la{solY}
Y_{v}=(A_{v}(x\!-vt\!-a),\Pi_{v}(x\!-vt\!-a), vt+a,  p_{v}),~~ |v|<1,~~ a\in\R^2,~~ p_{v}=mv+\langle A_{v},\rho\rangle.
\ee
Due to the first two equations of \eqref{mls3},  we get
\be\la{Hs5}
\Pi_{v}(y)=-(v\cdot\na)A_{v}(y), \qquad 
-(v\cdot\na)\Pi_{v}(y)=\De A_{v}(y)-\w J\varrho(y)+{\cal P}[v\rho(y)],
\ee
Hence, in the  Fourier representation,   
\be\la{solit3}
\hat A_{v}(k)=
\frac{(v-\frac{(v\cdot k)k}{k^2})\hat \rho(k)}{\hat D_0}-\frac{\w J\hat\varrho(k)}{\hat D_0},\quad \hat D_0:=k^2-(v\cdot k)^2,\quad  
{\w}=\frac{1}{I}\big(M+\langle A_{v}, J\varrho\rangle\big).
\ee
By \eqref{solit3},
\be\la{kappa}
\langle A_{v},J\varrho\rangle=-{\w}\langle \frac{J\hat\varrho(k)}{\hat D_0},J\hat\varrho(k)\rangle
=-{\w}\varkappa, \quad \varkappa=\varkappa(v)=\int\frac{|\hat\varrho(k)|^2}{\hat D_0}dk=\int\frac{|\na\hat\rho(k)|^2}{\hat D_0}dk.
\ee
Therefore,
\be\la{om-M}
{\w}={\w}(v)=M/(I+\varkappa)
\ee
By \eqref{zero1}, $\na A_{v}(y)\sim |y|^{-6}, \Pi_{v}(y)\sim |y|^{-6}$ as $|y|\to \infty$. 
Hence, 
\be\la{Pivom}
 \na A_{v},~ \Pi_{v}\in L^2_{\beta}(\R^2)\otimes\R^2, \quad \beta<5.
\ee
Denote   $V:=\{v\in\R^2:|v|<1\}$.
\begin{defin}
$S(\sigma):=(A_{v}(x-b),\Pi_{v}(x-b),b, p_{v})$, where $\sigma:=(b,v)$ with $b\in\R^2$ and $v\in V$.
\end{defin}
Obviously, the soliton solution \eqref{solY} admits the representation $S(\si(t))$ with $\si(t)=(vt+a,v)$.
\begin{defin}
A solitary manifold is the set ${\cal S}:=\{S(\si):\si\in\Sigma:=\R^2\times V\}$.
\end{defin}
Note that the manifold ${\cal S}$ is invariant with respect to the translations
 $$
 T_a:(A(\cdot),\Pi(\cdot),q,p)\mapsto(A(\cdot-a),\Pi(\cdot-a),q+a,p), \quad a\in\R^2.
 $$
 \vspace{-12mm}
\setcounter{equation}{0}
\section{Main result}
To state our main result, we introduce the weighted Sobolev spaces. Let  ${\bf L}^2_{\beta}$, $\beta\in\R$
be the  weighted space of solenoidal vector  field $A\in L^2_\beta$.
We define the space ${\cal F}_{\beta}$
of the  solenoidal vector  fields $(A,\Pi)$ with the norm
$
\Vert (A,\Pi)\Vert_{{\cal F}_\beta}=\Vert \na A\Vert_{{\bf L}^2_{\beta}}+\Vert \Pi\Vert_{{\bf L}^2_{\beta}}.
$
 Besides, we will use the space 
${\cal E}_{\beta}$ 
of the states $Y=(A,\Pi,q,p)$ with the norm
$$
\Vert Y\Vert_{\beta}=\Vert \na A\Vert_{{\bf L}^2_{\beta}}+\Vert \Pi\Vert_{{\bf L}^2_{\beta}}+|q|+|p|.
$$
\begin{theorem}\la{main}
Let conditions \eqref{rosym}--\eqref{zero1} and \eqref{M-condition}, and let $\beta=4+\delta$, $\delta\in (0,1)$.
Suppose that $Y_0\in {\cal E}_{\beta}$ is sufficiently close to the solitary manifold in the following sense:
\be\la{close}
Y_0=S(\si_0)+Z_0,\qquad d_\beta=\Vert Z_0\Vert_{\beta}\ll 1.
\ee
Let $Y(t)\in C(\R, {\cal E})$ be the solution to \eqref{mls3} with initial data $Y_0$. Then
\beqn\la{qq-as}
&&\dot q(t)= v_{\pm}+{\cal O}(|t|^{-2}),\quad  q(t)=v_{\pm}t+a_{\pm}+{\cal O}(|t|^{-2}),\\
\la{AP-as}
&&\left(\ba{cc} A(x,t)\\\Pi(x,t)\ea\right)=\left(
\ba{cc} A_{v_{\pm}}(x-v_{\pm}t-a_{\pm})\\ \Pi_{v_{\pm}}(x-v_{\pm}t-a_{\pm})\ea\right)
+W_0(t)\Psi_{\pm}+r_{\pm}(x,t),\quad t\to\pm\infty.
\eeqn
Here ${\w}_0={\w}(v_0)$, $W_0(t)$ is the dynamical group of the free wave equation, 
$\Psi_{\pm}$ are the corresponding asymptotic scattering states, 
and the remainder $r_{\pm}(x,t)$ converges to zero in the global energy norm:
\be\la{r-as}
\Vert r_{\pm}(t)\Vert_{\cal F}={\cal O}(|t|^{-1}),\quad t\to\pm\infty.
\ee
\end{theorem}
\setcounter{equation}{0}
\section{Symplectic Structure}
Consider the symplectic form $\Om$  on  the phase space ${\cal E}$,
\be\la{OmJ}
\Om(Y_1,Y_2)=\langle Y_1,{\bf J}Y_2\rangle,\quad Y_1,Y_2\in {\cal E},
\ee
where 
$
\langle Y_1, Y_2\rangle:=\langle A_1, A_2\rangle+
\langle\Pi_1,\Pi_2\rangle+q_1  q_2+p_1  p_2,
$
and $\langle A_1, A_2\rangle=\ds\int A_1(x) \ov A_2(x)dx$ etc.
\begin{defin}
i) $Y_1\nmid Y_2$   means that $Y_1$ is symplectic orthogonal to $Y_2$, i.e. $\Om(Y_1,Y_2)=0$.

ii) A projection operator ${\bf P}:{\cal E}\to{\cal E}$
is called symplectic
 orthogonal if $Y_1\nmid Y_2$ for $Y_1\in\mbox{\rm Ker}\5{\bf P}$ and
$Y_2\in\mbox {\rm Im}\5{\bf P}$.
\end{defin}
Denote by  ${\cal T}_{\si}{\cal S}$ to the tangent space to the manifold ${\cal S}$ at a point $S(\si)$.
The vectors
\be\la{inb}
\left.\ba{rclrclrclrclrclrcl}
\tau_j(v):=\pa_{b_j}S(\si)=&(-\pa_j A_{v}(y),&-\pa_j\Pi_{v}(y),&e_j& 0)
\\
\tau_{j+2}(v):=\pa_{v_j}S(\si)=&(\pa_{v_j}A_{v}(y),&\pa_{v_j}\Pi_{v}(y),& 0,& \pa_{v_j}p_{v})
\ea\right|, j=1,2,
\ee
form a basis in ${\cal T}_{\si}{\cal S}$. Here   $e_1=(1,0)$,  $e_2=(0,1)$.
\begin{lemma}\la{Ome}
The matrix ${\bf\Om}={\bf\Om}(v)$ with the elements $\Om(\tau_l(v),\tau_j(v))$ is non-degenerate for any $v\in V$.
\end{lemma}
We prove this lemma   in Appendix A.
The following lemma states that in a small neighborhood of the solitary manifold ${\cal S}$ a ``symplectic orthogonal projection''
onto ${\cal S}$ is well-defined. 
Denote 
$
v(Y)=\ds\frac{p-\langle A(x),\rho(x-y)\rangle}{m}.
$
\begin{lemma}\la{skewpro}
Let (\ref{rosym}) hold, $\al\in\R $ and $\ov v<1$.
Then
\\
i) there exists a neighborhood ${\cal O}_\al({\cal S})$ of ${\cal S}$ in ${\cal E}_\al$ and a
map ${\bf \Pi}:{\cal O}_\al({\cal S})\to{\cal S}$  such that ${\bf \Pi}$ is uniformly
continuous on ${\cal O}_\al({\cal S})\cap \{Y\in {\cal E}_\al: |v(Y)|\le {\ov v}\}$ in the metric of ${\cal E}_\al$,
\be\la{proj}
{\bf \Pi} Y=Y~~\mbox{for}~~ Y\in{\cal S}, ~~~~~\mbox{and}~~~~~Y-S \nmid {\cal T}_S{\cal S},~~\mbox{where}~~S={\bf \Pi} Y.
\ee
ii) ${\bf \Pi }T_aY=T_a{\bf \Pi} Y$ for $Y\in{\cal O}_{\al}({\cal S})$ and $a\in\R^2$.
\\
iii) For any $\ov v<1$ there exists a $\ti v<1$ s.t. $|v({\bf\Pi}Y)|<\ti v$ when $|v(Y)|<\ov v$.
\\
iv) For any $\ti v<1$ there exists a $z_\al(\ti v)>0$ s.t. $S(\si)+Z\in{\cal O}_\al({\cal S})$
if $|v(S(\si))|<\ti v$ and $\Vert Z\Vert_\al<z_\al(\ti v)$.
\end{lemma}
The proof is similar to that of  Lemma 3.4 in \cite{IKV2012}. We will call ${\bf \Pi}$ the symplectic orthogonal projection onto ${\cal S}$.
\begin{cor}
The condition (\ref{close}) implies that we can assume $Y_0=S+Z_0$ where $S=S(\si_0)={\bf \Pi} Y_0$, and
\be\la{closeZ}
\Vert Z_0\Vert_{\beta} \ll 1.
\ee
\end{cor}
\setcounter{equation}{0}
\section{Linearization on the solitary manifold}
Let us consider a solution to the system (\ref{mls3}), and split it as  the sum
\be\la{dec}
Y(t)=S(\si(t))+Z(t),
\ee
where $\si(t)=(b(t),v(t))\in\Sigma$ is an arbitrary smooth function of $t\in\R$.
In detail, denote $Y=(A,\Pi,q,p)$ and $Z=(\Lambda,\Psi,r,\pi)$.
Then (\ref{dec}) means that
\be\la{add}
\left.
\ba{rclr}
A(x,t)&=&A_{v(t)}(x-b(t))+\Lambda(x-b(t),t), ~~ \Pi(x,t)=\Pi_{v(t)}(x-b(t))+\Psi(x-b(t),t)\\
q(t)&=&b(t)+r(t),~~~~~p(t)=p_{v}(t)+\pi(t)
\ea
\right|
\ee
Substituting  (\ref{add}) to (\ref{mls3}) and setting $y=x-b(t)$, we obtain 
\beqn\la{dot-A}
\dot A&=&\dot v\cdot \na_v A_{v(t)}(y)-\dot b\cdot \na A_{v(t)}(y)+
\dot\Lambda(y,t)-\dot b\cdot \na \Lambda (y,t)=\Pi_{v(t)}(y)+\Psi(y,t)\\
\nonumber
\dot\Pi&=&\dot v\cdot \na_v\Pi_{v(t)}(y)-\dot b \cdot\na\Pi_{v(t)}(y)+\dot\Psi(y,t)-\dot b\cdot \na\Psi(y,t)=\De A_{v(t)}(y)+\De\Lambda(y,t)\\
\la{dot-Pi}
&-&\frac{1}{I}\big(M+\langle A_{v(t)}(y)+\Lam(y,t),J\varrho(y-r)\rangle\big)J\varrho(y-r)+{\cal P}[\dot q \rho(y-r)]\\
\la{dot-q}
m\dot q&=&m(\dot b+\dot r)=p_{v}+\pi-\langle A_{v(t)}+\Lambda,\rho(y-r)\rangle\\
\nonumber
\dot p&=&\dot v\cdot \na_v p_{v(t)}+\dot \pi(t)
=\langle \na((A_{v(t)}(y)+\Lambda(y,t))\cdot\dot q),\rho(y-r)\rangle\\  
\la{dot-p}
&+&\frac{1}{I}\big(M+\langle A_{v(t)}(y)+\Lam(y,t),J\varrho(y-r)\rangle\big)\langle\na\cdot J(A_{v(t)}(y)+\Lambda(y,t)),\varrho(y-r)\rangle
\eeqn
{\it i)} Equations  (\ref{Hs5}) and \eqref{dot-A}  imply
\be\la{Lambda-eq}
\dot \Lambda(y,t)=\Psi(y,t)+\dot b\cdot \na \Lambda(y,t)+
(\dot b-v)\cdot \na A_{v}(y)-\dot v\cdot \na_v A_{v}(y).
\ee
{\it ii)}  Equations (\ref{Hs5}) and  \eqref{dot-Pi} imply
\beqn\nonumber
&&\dot \Psi(y,t)=\De\Lambda(y,t)+\dot b\cdot \!\na\Psi(y,t)-\dot v\cdot \na_v \Pi_{v}(y)+(\dot b-v)\cdot\! \na\Pi_{v}(y)+{\cal P}[\dot q \rho(y-r)-v\rho(y)]\\
\la{dotPsi}
&&-\frac{1}{I}\Big(\big(M+\langle A_{v}(y)+\Lam(y,t),J\varrho(y-r)\rangle\big) J\varrho(y-r)-\big(M+\langle A_{v}(y),J\varrho(y)\rangle\big) J\varrho(y)\Big)
\eeqn
First,  note that
\be\la{rho-exp}
\rho(y-r)=\rho(y)-r \cdot\na\rho(y)+N_{\rho}(r,y),\quad \varrho(y-r)=\varrho(y)-r \cdot\na\varrho(y)+N_{\varrho}(v,r),
\ee
where
$$
\Vert N_{\rho}(v,r)\Vert_\beta ={\cal R}_{v}(r^2),\quad   \Vert N_{\varrho}(v,r)\Vert_\beta ={\cal R}_{v} (r^2),
$$
Here ${\cal R}_{v}(A)$ is a general notation for a positive function which remains bounded as $A$ is sufficiently small.
Note that this  bound  is uniform  in $v$ and $Z$ with $\Vert Z\Vert_{-\beta}\le z_{-\beta}(\ti v)$ and $|v|<\ti v<1$.
Further,  \eqref{solY}, \eqref{dot-q} and  \eqref{rho-exp}  imply
\beqn\nonumber
m\dot q&=&mv+\langle  A_v,\rho\rangle+\pi-\langle A_v+\Lambda,\rho(y-r)\rangle\\
\la{dot-q-exp}
&=&mv+\pi-\langle  \Lam,\rho\rangle +\langle A_v, r\cdot\na\rho\rangle+\langle \Lambda, r\cdot\na\rho\rangle-\langle A_v+\Lam, N_{\rho}(v,r)\rangle.
\eeqn
By \eqref{solit3},
\be\la{A-r-rho}
\langle A_v, r\cdot\na\rho\rangle=-{\w}\langle  \frac{J\na\hat\rho(k)}{\hat D_0},(r\cdot k)\hat\rho(k)\rangle=-{\w} JPr,
\ee
where 
 \be\la{c-ro}
P=\{P_{jl}\}, ~~{\rm with}~~P_{lj}=P_{jl}:= \int \frac{k_l\na_j\hat\rho(k)}{\hat D_0}\hat\rho(k) dk,\quad j,l=1,2.
 \ee
Hence, \eqref{rho-exp}--\eqref{dot-q-exp} imply that
\be\la{qr}
\dot q\rho(y-r)-v\rho(y)
=\frac {\rho(y)}m \big(\pi-\langle \Lam,\rho\rangle-{\w} JPr\big)-vr \cdot\na\rho(y)+N'(v,Z),~~\Vert N'(v,Z)\Vert_\beta={\cal R}_{\ti v}\Vert Z\Vert^2_{-\beta}.
\ee
Further, \eqref{solit3} implies
\be\la{Acdot}
\langle A_v, (r\cdot\na) J\varrho\rangle=
=-\int\frac{(v\cdot Jk)(r\cdot k)(k\cdot\na\hat\rho)\hat\rho}{k^2\hat D_0}dk=-\int\frac{(v\cdot Jk)(r\cdot\na\hat\rho)\hat\rho}{\hat D_0}dk=r\cdot PJv.
\ee
since
$$
(r\cdot k)(k\cdot\na\hat\rho)=(r_1k_1+r_2k_2)(k_1\na_1\hat\rho+k_2\na_2\hat\rho)
=r_1\na_1\hat\rho(k_1^2+k_2^2) +r_2\na_2\hat\rho(k_1^2+k_2^2)=k^2(r\cdot\na\hat\rho)
$$
Therefore,  
\beqn\nonumber
&&\!\!\!\!\!\!\!\!\!\!\!\!\!\!\!\!\!\!\!\!(M+\langle A_v+\Lam,J\varrho(y-r)\rangle) J\varrho(y-r)-(M+\langle A_v,J\varrho\rangle)J\varrho(y)\\
\nonumber
&&\!\!\!\!\!\!\!\!\!\!\!\!\!\!\!\!\!\!\!\!=-I{\w}  (r\cdot\na)J\varrho(y)-(r\cdot PJv)J\varrho(y)+\langle \Lam,J\varrho\rangle J\varrho(y)+N''(v, Z),
~~\Vert N''(v,Z)\Vert_\beta={\cal R}_{v}\Vert Z\Vert^2_{-\beta}
\eeqn
by \eqref{rho-exp} and \eqref{Acdot}. 
Substituting  this and  \eqref{qr} into \eqref{dotPsi}, we obtain
\beqn\la{Psi-eq}
\!\!\!\!\!\!\!\!\!\!\!\!\!\!\!\!\!\!\!\!&&\dot \Psi=\De\Lambda+\dot b\cdot \!\na\Psi\!+(\dot b\!-\!v)\!\cdot\! \na\Pi_{v}\!
-\dot v\cdot \na_v \Pi_{v}+{\cal P}\Big[\frac{\rho}m\big(\pi-{\w} JPr-\langle \Lam,\rho\rangle\big)-vr \cdot\na\rho\Big]\\
\nonumber
\!\!\!\!\!\!\!\!\!\!\!\!\!\!\!\!\!\!\!\!&&
+{\w}  (r\cdot\na)J\varrho-\frac{1}{I}\big(\langle \Lam,J\varrho\rangle-r\cdot PJv\big)J\varrho+N_2(v,Z),
\quad \Vert N_2(v,Z)\Vert_\beta={\cal R}_{v}\Vert Z\Vert^2_{-\beta}.
\eeqn
\smallskip\\
{\it iii)} 
Equations  (\ref{solY}), (\ref{dot-q}) and \eqref{A-r-rho} imply
\be\la{r-eq}
m\dot r=-m\dot b+mv+\langle A_v,\rho\rangle+\pi-\langle A_v+\Lambda,\rho(y-r)\rangle
=-m(\dot b-v)+\pi-\langle\Lam,\rho\rangle-{\w} JPr+N_3(v,Z),
\ee
where $\Vert N_3(v,Z)\Vert_\beta={\cal R}_{v}\Vert Z\Vert^2_{-\beta}$.
\smallskip\\
{\it iv)} 
Using \eqref{dot-p} and the last equation \eqref{mls3} for soliton, we obtain
\beqn\nonumber
&&\!\!\!\!\!\!\!\!\!\!\!\!\!\!\!\!\!\!\!\!\dot \pi=-\dot v\cdot \na_v p_{v}+\langle \na((A_v+\Lambda)\cdot\dot q),\rho(y-r)\rangle-\langle \nabla(A_v\cdot v),\rho\rangle\\
\la{dot-pi}
&&\!\!\!\!\!\!\!\!\!\!\!\!\!\!\!\!\!\!\!\!+\frac{1}{I}\big(M+\langle A_v+\Lam,J\varrho(y-r)\rangle\big)\langle\na\cdot J(A_v+\Lambda,\varrho(y-r)\rangle
-\frac{1}{I}\big(M+\langle A_v,J\varrho\rangle\big)\langle \nabla\cdot JA_v,\varrho\rangle\big)
\eeqn
First, note that 
\be\la{APJ}
\langle\na(A_v\cdot\pi),\rho\rangle=-\langle \pi\cdot A_v,\na\rho\rangle
={\w}\int\frac{\pi\cdot J\na\hat\rho}{\hat D_0}k\hat\rho dk=-{\w} PJ\pi.
\ee
Similarly,
\be\la{AJJP}
\langle\na(A_v\cdot\langle\Lambda,\rho\rangle),\rho\rangle=-{\w} PJ\langle \Lam,\rho\rangle,\qquad
\langle \na (A_v\cdot  JPr),\rho\rangle =-{\w} PJJPr={\w} P^2r.
\ee
Moreover,
\be\la{AJJP+}
\langle \na(A_v\cdot v), r\cdot\na\rho\rangle=\int k\frac {v^2k^2-(v\cdot k)^2}{k^2\hat D_0}(r\cdot k)|\hat\rho(k)|^2dk=Qr,
\ee
where we denote
\be\la{h-ro}
Q=\{Q_{jl}\}, \quad Q_{jl}:=\int \big(v^2k^2-(v\cdot k)^2\big)\frac{k_lk_j|\hat\rho(k)|^2dk}{k^2\hat D_0},\quad j,l=1,2.
\ee
Therefore, using \eqref{rho-exp}--\eqref{A-r-rho} and \eqref{AJJP}--\eqref{AJJP+}, we get 
\beqn\nonumber
\!\!\!\!\!\!\!\!\!\!\!\!\!\!&&\langle \na((A_v+\Lambda)\cdot\dot q),\rho(y-r)\rangle-\langle \na(A_v\cdot v),\rho\rangle\\
\la{s1}
\!\!\!\!\!\!\!\!\!\!\!\!\!\!&&=-\langle v\cdot \!\Lam,\na\rho\rangle-\!\frac {\w}{m}PJ(\pi\!-\!\langle \Lam,\rho\rangle)-(\frac {\w^2}{m}P^2\!+Q)r\!+\!N'''(v,\!Z),~~
\Vert N'''(v,\!Z)\Vert_\beta=\!{\cal R}_{v}\Vert Z\Vert^2_{-\beta}.
\eeqn
Further, we notice that
\beqn\nonumber
\langle\na\cdot JA_v,\varrho\rangle&=&\langle k\cdot J\hat A_v,\na\hat\varrho\rangle=-\int \frac{(v\cdot Jk)\hat\rho}{\hat D_0}\,\na\hat\rho dk=PJ v,\\
\la{JAro}
\langle\na\cdot JA_v,r\cdot\na\varrho\rangle
&=&{\w}\int \frac{(k\cdot\na\hat\rho)}{\hat D_0}(r\cdot k) \na\rho dk={\w} Fr, 
\eeqn
where we denote
\be\la{f-ro}
F=\{F_{jl}\}~~~{\rm with}~~F_{jl}:=\int \frac{(k\cdot \na\hat\rho)k_j\na_l\hat\rho}{\hat D_0} dk:=\int \frac{k_jk_l|\na\hat\rho|^2}{\hat D_0} dk..
\ee
Therefore, \eqref{rho-exp}, \eqref{Acdot} and \eqref{JAro} impliy
\beqn\nonumber
&&\frac{1}{I}\big(M+\langle A_v(y)+\Lam(y,t),J\varrho(y-r)\rangle\big)\langle\na\cdot J(A_v+\Lambda),\varrho(y-r)\rangle
-\frac{1}{I}(M+\langle A_v,J\varrho\rangle\big)\langle\na\cdot JA_v,\varrho\rangle\\
\la{s2}
&&={\w}\langle \na\cdot J\Lam,\varrho\rangle-{\w}^2 Fr+ \frac{1}{I}\langle \Lam,J\varrho\rangle PJv
-\frac 1{I}(r\cdot PJv) PJv+N^{\rm iv}(v,Z).
\eeqn
where $\Vert N^{\rm iv}(v,Z)\Vert_\beta={\cal R}_{v}\Vert Z\Vert^2_{-\beta}$. Substituting \eqref{s1} and \eqref{s2} into \eqref{dot-pi}, we get
\beqn\nonumber
\dot \pi&\!\!\!=\!\!\!&-\dot v\cdot \na_v p_{v}-\langle v\cdot\Lambda,\na\rho\rangle-\frac {\w}{m} PJ\pi+\frac {\w}m PJ\langle \Lam,\rho\rangle
+{\w}\langle \na\cdot J\Lam,\varrho\rangle-(\frac {{\w}^2 }m P^2\!+Q+{\w}^2 F)r\\
\la{pi-eq}  
&\!\!\!-\!\!\!&\frac 1{I}\Big(\!(r\cdot PJv)-\langle \Lam,J\varrho\rangle\!\Big) PJv+N_4(v,Z),\quad \Vert N_4(v,Z)\Vert_\beta={\cal R}_v\Vert Z\Vert^2_{-\beta}.
\eeqn
\smallskip\\
Combining equations \eqref{Lambda-eq}, \eqref{Psi-eq}, \eqref{r-eq}, and \eqref{pi-eq} together, we obtain
\be\la{lin}
\dot Z(t)={\bf A}(t)Z(t)+T(t)+N(t),\,\,\,t\in\R, \quad Z=(\Lam,\Psi,r,\pi).
\ee
Here the operator ${\bf A}(t)={\bf A}_{v(t), u(t)}$ depends on two parameters, $v=v(t)$, 
and $u(t):=\dot b(t)$. It  can be written in the matrix form
\be\la{AA}
{\bf A}:=\left(
\ba{cccccc}
u \cdot\na & 1 & 0 & 0 \\
\De -{\cal P}[\frac{\langle\cdot,\rho\rangle\rho}{m}] -\frac{\langle \cdot,J\varrho\rangle J\varrho}{I}& u \cdot \na & {B}&{\cal P} [\frac{\rho}m\cdot] \\
-\frac 1m\langle\cdot,\rho\rangle & 0 & -\frac{\w}{m}JP & \frac 1m\\
  {B}^* & 0 &-S & -\frac{\w}mPJ&\\
\ea\right)
\ee
where we denote
\beqn \la{B-op+}
 {B}r:&\!\!=\!\!&{\cal P}[-\frac{\w}{m}\rho JPr-v(r\cdot \na\rho)]+{\w}(r\cdot \na)J\varrho+\frac{1}{I}(r\cdot PJv)J\varrho,\\
 \la{B-op}
  {B}^*\Lam:&\!\!=\!\!&\frac {\w}m PJ\langle \Lam,\rho\rangle-\langle v\cdot\Lam,\na\rho\rangle+{\w}\langle \na\cdot J\Lam,\varrho\rangle
  +\frac{1}{I}PJv\langle\Lam,J\varrho\rangle,\\
 \la{j-al}
Sr:&\!\!=\!\!&(\frac {\om^2 }m P^2+Q+\om^2 F)r+\frac 1{I}(r\cdot PJv) PJv.
\eeqn
Furthermore,   $T(t)=T_{v(t), u(t)}$ and $N(t)=N(v(t),Z(t))$ in  (\ref{lin}) stand for
\be\la{T}
T(t)=\left(
\ba{c}
(u-v)\cdot \na A_{v}-\dot v\cdot \na_v A_{v}\\
(u-v)\!\cdot\na\Pi_{v}-\dot v\cdot \na_v \Pi_{v}\\
v-u \\
-\dot v\cdot \na_v p_{v}\\
\ea
\right),\quad
N(t)=\left(
\ba{c}
0 \\ N_2(,vZ)  \\\ N_3(v,Z) \\ N_4 (v,Z)
\ea
\right).
\ee
The remainder term $N(v,Z)$ satisfies the estimate
\be\la{N-est}
\Vert N(v,Z)\Vert_\beta\le C(\ti v)\Vert Z\Vert^2_{-\beta}.
\ee
uniformly in $v$ and $Z$ with $\Vert Z\Vert_{-\beta}\le z_{-\beta}(\ti v)$ and $|v|<\ti v<1$.
 \setcounter{equation}{0}
\section{The Linearized Equation}
Here we collect some Hamiltonian and spectral properties of the
generator (\ref{AA}) of the linearized equation
\be\la{line}
\dot X(t)={\bf A}_{v,u}X(t),~~~~~~~t\in\R
\ee
with arbitrary fixed $v\in V=\{v\in\R^2: |v|<1\}$ and $u\in \R^2$
\begin{lemma} \la{haml}
i) The equation (\ref{line}) formally can be written as the Hamilton system,
\be\la{lineh}
\dot X(t)=
{\bf J}D{\cal H}_{v,u}(X(t)),~~~~~~~t\in\R,
\ee
where $D{\cal H}_{v,u}$ is the Fr\'echet derivative of the Hamilton functional
\beqn\nonumber
{\cal H}_{v,u}(X)&=&\frac 12\int\Big[|\Psi|^2+|\na\Lambda|^2\Big]dy+\int\Psi (u\cdot\na)\Lambda dy
+\frac{1}{2m}\langle\Lambda,\rho\rangle^2+\frac {1}{2I}\langle \Lambda,J\varrho\rangle^2\\
\nonumber
&+&\int (\frac{\w}{m}JP r-\frac{\pi}m)\cdot\Lambda\rho dy-\frac{1}{I}r\cdot PJv\int \Lambda\cdot J\varrho dy+\int(\na\rho\cdot r)v\cdot \Lambda dy\\
\la{H0}
&-&{\w}\int(r\cdot\na)\Lam\cdot J\varrho dy +\frac{\pi^2}{2m} -\frac{\w}{m} JPr \cdot\pi+\frac 12 r\cdot Sr,\quad X=(\Lambda,\Psi,r,\pi)\in \cE.
\eeqn
ii) Energy conservation law holds for  solutions $X(t)\in C^1(\R,\cE)$,
\be\la{enec}
{\cal H}_{v,u}(X(t))={\rm const},~~~~~t\in\R.
\ee
iii) The skew-symmetry relation holds,
\be\la{com}
\Omega({\bf A}_{v,u}X_1,X_2)=-\Omega(X_1,{\bf A}_{v,u}X_2), ~~~~~~~~X_1,X_2\in \cE.
\ee
\end{lemma}
\begin{proof}
i) Equality \eqref{lineh}   is easily verified by differentiation, see also (\ref{canH}).\\
ii) The energy conservation law follows by  (\ref{lineh}) and the chain rule for the Fr\'echet derivatives. Formally,
$\ds\frac d{dt}\cH_{v,u}(X(t))=\langle D\cH_{v,u}(X(t)),\dot X(t)\rangle=
\langle D\cH_{v,u}(X(t)),{\bf J} D\cH_{v,u}(X(t))\rangle=0$ for  $t\in\R$,
since the operator ${\bf J}$ is skew-symmetric, and $D\cH_{v,u}(X(t))\in \cE$ for  $X(t)\in \cE$.
\\
iii) The skew-symmetry holds since ${\bf A}_{v,u}X={\bf J}D\cH_{v,u}(X)$, and the linear operator
$X\mapsto D\cH_{v,u}(X)$ is symmetric as the Fr\'echet derivative of a quadratic form.
\end{proof}
\begin{lemma} \la{ljf}
The operator ${\bf A}_{v,u}$ acts on the tangent vectors $\tau_j(v)$ to the solitary manifold as follows,
\be\la{Atan}
{\bf A}_{v,u}[\tau_j(v)]=(u-v)\cdot\na\tau_j(v),\quad
{\bf A}_{v,u}[\tau_{j+2}(v)]=(u-v)\cdot\na\tau_{j+2}(v)+\tau_j(v,),\qquad j=1,2.
\ee
\end{lemma}
We prove the lemma in Appendix B. 
\begin{cor}\la{ceig1}
In the case  $u=v$, the tangent vectors $\tau_j(v)$ are eigenvectors,
and $\tau_{j+2}(v)$ are root vectors of the operator ${\bf A}_{v}:={\bf A}_{v,v}$, corresponding to zero eigenvalue, i.e.
\be\la{Atanformv}
{\bf A}_{v}[\tau_j(v)]=0, \quad  {\bf A}_{v}[\tau_{j+2}(v)]=\tau_j(v),\quad j=1,2.
\ee
\end{cor}
 \begin{lemma}\la{ceig}
The Hamilton function ${\cal H}_{v}(X):={\cal H}_{v,v}(X)$ is nonnegative definite,
\be\la{H+}
{\cal H}_{v}(X)\ge 0,\quad X\in {\cal E}.
\ee
\end{lemma}
\begin{proof}
We can assume that $v = (|v|, 0)$. We split  ${\cal H}_{v}(X)$ as
\be\la{H++}
{\cal H}_{v}(X)=\frac12 \Vert\Psi+(v\cdot\na)\Lam\Vert_{L^2(\R^2)}^2+\frac {1}{2I}(r\cdot PJv-\langle \Lambda,J\varrho\rangle)^2
+\frac{1}{2m}(\pi-\langle\Lam,\rho\rangle-{\w} JP r)^2+h_{v}(X)
\ee
It remains to prove that
\be\la{H4}
h_{v}(X)=\frac 12\langle(-\Delta+(v\cdot\na)^2\Lam,\Lam\rangle
+\!\int\!\big((\na\rho\cdot r)v\cdot \Lambda-{\w}(r\cdot\na)\Lam\cdot J\varrho\big) dy
+\sum\limits_{j}\frac {v^2q_j+{\w}^2 f_j}{2}r_j^2\ge 0.
\ee
 In Fourier space, $h_{v}(X)$ reads
\beqn\nonumber
h_{v}(X)&=&\frac 12\int\Big(\hat D_0|\hat\Lam|^2-2i(k\cdot r)\hat\rho |v|\hat\Lam_1-2{\w}(r\cdot k)\hat\Lam\cdot J\na\hat\rho
 \Big)dk
+\frac 12\sum\limits_{j} (v^2q_j+{\w}^2 f_j)r_j^2\\
\nonumber
&=&\frac 12\int\big(\hat D_0|\rIm\hat\Lam|^2+2(k\cdot r)\hat\rho |v|\rIm\hat\Lam_1\big)dk+\frac{v^2}2\sum\limits_{j}q_jr_j^2\\
\nonumber
&+&\frac 12\int\big(\hat D_0|\rRe\hat\Lam|^2-2{\w}(k\cdot r)\rRe\hat\Lam\cdot J\na\rho \Big)dk
+\frac{{\w}^2}2 \sum\limits_{j}f_jr_j^2.
\eeqn
We used here the fact that   $\rRe\hat\Lam_j$ is even and $\rIm\hat\Lam_j$ is odd because $\Lam_j\in \R$.
By \eqref{h-ro}, 
\beqn\nonumber
&&\int\big(\hat D_0|\rIm\hat\Lam|^2+2(k\cdot r)\hat\rho |v|\rIm\hat\Lam_1\big)dk+v^2\sum\limits_{j}q_jr_j^2\\
\nonumber
&&=\int\Big(\hat D_0|\rIm\hat\Lam|^2+2(k\cdot r)\hat\rho |v|(e_1-\frac{k_1k}{k^2})\cdot\rIm\hat\Lam
+v^2\frac{(k\cdot r)^2|\hat\rho|^2}{\hat D_0}(e_1-\frac{k_1k}{k^2})^2\Big)dk\\
\nonumber
&&=\int \Big(\hat D_0^{1/2}\rIm\hat\Lam+\frac{(k\cdot r)|v|\hat\rho}{\hat D_0^{1/2}}(e_1-\frac{k_1k}{k^2})\Big)^2dk\ge 0.
\eeqn
since $\nabla\cdot\Lam=0$. 
Finally, \eqref{f-ro} implies
\beqn\nonumber
&&\int\big(\hat D_0|\rRe\hat\Lam|^2-2{\w}(k\cdot r)\rRe\hat\Lam\cdot J\na\rho\big)dk +{\w}^2 \sum\limits_{j}f_jr_j^2\\
\nonumber
&&=\int\Big(\hat D_0|\rRe\hat\Lam|^2+{2\w} (k\cdot r)(\rRe\hat\Lam_1\na_2\hat\rho-\rRe\hat\Lam_2\na_1\hat\rho)
+{\w}^2\frac{(k\cdot r)^2}{\hat D_0}\big((\na_1\hat\rho)^2+(\na_2\hat\rho)^2\big)\Big)dk\\
\nonumber
&&=\int\Big(\hat D_0^{1/2}\rRe\hat\Lam_1-\frac{\om(k\cdot r)\na_2\hat\rho}{\hat D_0^{1/2}}\Big)^2dk+
\int\Big(\hat D_0^{1/2}\rRe\hat\Lam_2+\frac{{\w}(k\cdot r)\na_1\hat\rho}{\hat D_0^{1/2}}\Big)^2dk\ge 0.
\eeqn
\end{proof}
\setcounter{equation}{0}
\section{Symplectic Decomposition of the Dynamics}
Now we are going to choose $S(\si(t)):={\bf \Pi} Y(t)$ in (\ref{dec}). This is possible for $t=0$ by (\ref{close}),  
so $S(\si(0))={\bf\Pi} Y(0)$ and  $Z(0)=Y(0)-S(\si(0))$ are well defined.
Moreover, a priori estimate (\ref{apr}) and Lemma \ref{skewpro} 
imply  that $S(\si(t))={\bf\Pi} Y(t)$ and  $Z(t)=Y(t)-S(\si(t))$ are well defined for $t\ge 0$ until
 $\Vert Z(t)\Vert_{-\beta} < z_{-\beta}(\ti v)$. This is formalized by the following  definition.
\begin{defin}
$t_*$ is the ``exit time'',
\be\la{t*}
t_*=\sup \{t>0: \Vert Z(s)\Vert_{-\beta} < z_{-\beta}(\ti v),~~0\le s\le t\}.
\ee
\end{defin}
For $0<t<t_*$, we  set $S(\si(t))={\bf \Pi} Y(t)$ which is equivalent to the symplectic orthogonality condition  (\ref{proj}):
\be\la{orth}
\Om(Z(t),\tau_j(t))=0,\quad \tau_j(t)=\tau_j(\si(t)),\quad j=1,\dots,4, ~~~~~~~0\le t<t_*. 
\ee
Denote 
\be\la{vw}
c(t):=b(t)-\ds\int^t_0 v(\tau)d\tau,\quad \dot c(t)=\dot b(t)-v(t)=u(t)-v(t), \quad 0\le t<t_*.
\ee
\begin{lemma}\la{mod}
Let $Y(t)$ be a solution to the Cauchy problem for \eqref{mls3}, and (\ref{dec}), (\ref{orth}) hold. 
Then 
\be\la{parameq}
\dot c(t), \,\dot v(t) ={\cal O}(\Vert Z\Vert_{-\beta}^2), ~~~~~~~0\le t<t_*,
\ee
uniformly in $v$ and $Z$ with $\Vert Z\Vert_{-\beta}\le z_{-\beta}(\ti v)$ and $|v|<\ti v<1$.
\end{lemma}
The proof is similar to that of  Lemma 10.1 in \cite{K2025}. 
Now we rewrite (\ref{lin}) as
\be\la{reduced}
\dot Z(t)={\bf A}_{v(t),u(t)}Z(t)+\ti N(t), ~~~~~~~~~0\le t<t_*,
\ee
where $\ti N(t):=T(t)+N(t)$.
Formula (\ref{T}) and Lemma \ref{mod}  imply that
$\Vert T(t)\Vert_{\beta}\le C(\ti v)\Vert Z\Vert_{-\beta}^2$ for $0\le t<t_*$,
Hence,  
\be\la{redN}
\Vert\ti  N(t)\Vert_{\beta}\le C(\ti v)\Vert Z\Vert_{-\beta}^2,\qquad 0\le t<t_*.
\ee
\setcounter{equation}{0}
\section{Frozen Transversal Dynamics}
Now  we fix  arbitrary $t_1\in [0,t_*)$, and rewrite the equation (\ref{reduced}) in a ``frozen form''
\be\la{froz}
\dot Z(t)={\bf A}_1Z(t)+({\bf A}_{v(t),u(t)}-{\bf A}_1)Z(t)+\ti N(t),\,\,\,~~~~0\le t<t_*,
\ee
where ${\bf A}_1={\bf A}_{v(t_1),v(t_1)}$. The next trick allows us to kill the ``bad terms"  $[u(t)-v(t_1)] \cdot\na$ in the upper left corner of the
matrix $A_{v(t),u(t)} -A_1$. We denote
\be\la{dd1}
\ell_1(t):=\int_{t_1}^t(u(s)-v(t_1))ds, ~~~~0\le t\le t_1.
\ee
 and  change the  variables $(y,t)\mapsto (y_1,t)=(y+d_1(t),t)$.
Next we define 
\beqn\la{Z1}
\Lam^+(y_1,t):=\Lam(y,t)=\Lam(y_1-\ell_1(t),t),\quad \Psi^+(y_1,t):=\Psi(y,t)=\Psi(y_1-\ell_1(t),t)). 
\eeqn
Then we obtain the final form of the ``frozen equation'' for the transversal dynamics
\be\la{redy1}
\dot Z^+(t)={\bf A}_1Z_1(t)+{\bf B}(t)Z^+(t)+\ti N^+(t),\quad Z^+(t)=(\Lam^+(t),\Psi^+(t),r(t),\pi(t)),\quad0\le t\le t_1,
\ee
where $\ti N_1(t)$ expressed in terms of $y=y_1-d_1(t)$,  
and nonzero entries of the matrix  ${\bf B}(t)$  acting on $Z^+(t)$ as follows
\beqn\nonumber
B_{1}(t)\Lam\!\!&\!\!\!\!=\!\!\!\!&\!\!({\w}(t)\!-\!{\w}_1)\Big(\frac{1}{m}PJ\langle\Lambda,\rho\rangle
+\langle\na\cdot J\Lambda,\varrho\rangle\Big)+ \langle (v_1\!-\!v(t))\cdot\!\Lam,\na\rho\rangle
+\frac{PJ(v(t)\!-\!v_1))}{I}\langle\Lam,J\varrho\rangle,\\
\nonumber
B_{2}(t)r\!\!&\!\!\!\!=\!\!\!\!&\!\! {\cal P}[\frac{{\w}_1\!-\!{\w}(t)}{m}\rho JPr+(v_1\!-\!v(t))r\cdot\na\rho)]+({\w}(t)\!-\!{\w}_1)(r\cdot\na)J\varrho
+\frac{1}{I}(r\cdot PJ(v(t)-v_1))J\varrho,\\
\nonumber 
B_{3}(t)r\!\!&\!\!\!\!=\!\!\!\!&\!\!\Big(Q_1-Q(t)+{\w}_1F_1-{\w}(t)F(t)+\frac{1}{m}\big({\w}_1^2P^2_1-{\w}^2(t)P^2(t)\big)\Big)r\\
\nonumber
&+&\frac{1}{I}\Big((r\cdot PJv_1)PJv_1-(r\cdot PJv(t))PJv(t)\Big),\\
\nonumber
B_{4}(t)\pi\!\!&\!\!\!\!=\!\!\!\!&\!\!\frac{1}{m}\big({\w}(t)-{\w}_1)PJ\pi.
\eeqn
Denote $\ov \ell_1(s):=\sup_{0\le t\le s} |\ell_1(t)| $, $0\le s\le t_1$ and reduce the exit time:
\be\la{t*'}
t_*'=\sup \{t\in[0,t_*):
\ov \ell_1(s)\le 1,~~0\le s\le t\}.
\ee
Using \eqref{c-ro}, \eqref{h-ro}, \eqref{f-ro} and \eqref{parameq}, we obtain for $t_1 < t_*'$ :
\be\la{B1Z1est}
\Vert{\bf B}(t)Z^+(t)\Vert_{\beta}\le  C\Vert Z(t)\Vert_{-\beta}|v(t)-v_1|
\le  C(\ti v)\Vert Z(t)\Vert_{-\beta}\int_{t_1}^{t}  \Vert Z(s)\Vert_{-\beta}^2 ds, \quad  0\le t\le t_1.
\ee
Similarly,  
\be\la{N1est}
\Vert\ti N^+(t)\Vert_{\beta}\le C(\ti v) \Vert Z(t)\Vert_{-\beta}^2, \quad  0\le t\le t_1<t_*'.
\ee
\setcounter{equation}{0}
\section{Decay of linearized dynamics}
Here we will analyze  the  linear equation,
\be\la{Avv}
\dot X(t)={\bf A}_1X(t), ~~t\in\R,
\ee
where ${\bf A}_1={\bf A}_{v_1,v_1}$ with $v_1=v(t_1)$,  and a fixed $t_1\in[0,t_*')$.
The solutions to this equation does not decay
without  the orthogonality condition  of type (\ref{orth}) since the equation  admits the {\it secular solutions}
$$
X(t)=\sum\limits_{j=1,2}C_{j}\tau_j(v)+\sum\limits_{j=1,2}D_j[\tau_1(v)t+\tau_{j+2}(v,\om)]. 
$$
We will apply the  symplectic orthogonal projection which kills such solutions.
\begin{defin}
i)  For $v\in V$, denote by ${\bf\Pi}_{v}$ the symplectic orthogonal projection
of ${\cal E}$ onto the tangent space ${\cal T}_{S(\si)}{\cal S}$, and  ${\bf P}_{v}:={\bf I}-{\bf\Pi}_{v}$.
\\
ii) Denote by ${\cal Z}_{v}={\bf P}_{v}{\cal E}$ the space symplectic orthogonal to ${\cal T}_{S(\si)}{\cal S}$.
\end{defin}
\begin{remark}
Note that, by the linearity, 
\be\la{Piv}
{\bf\Pi}_vZ=\sum{\bf\Pi}_{jl}(v)\tau_j(v)\Om(\tau_l(v),Z),\quad Z\in{\cal E},
\ee
with some smooth coefficients ${\bf\Pi}_{jl}(v)$. Hence, the projector ${\bf\Pi}_v$, in the variable $y=x-b$,
does not depend on $b$.
\end{remark}
Now we have the symplectic orthogonal decomposition
\be\la{sod}
{\cal E}={\cal T}_{S(\si)}{\cal S}+{\cal Z}_{v},~~~~~~~\si=(b,v), 
\ee
and the symplectic orthogonality  (\ref{orth})
can be written in the following equivalent forms,
\be\la{PZ}
{\bf\Pi}_{v(t)} Z(t)=0,~~~~{\bf P}_{v(t)}Z(t)= Z(t),~~~~~~~~~0\le t<t_*'.
\ee
\begin{remark}r\la{rZ}
{\rm
The tangent space ${\cal T}_{S(\si)}{\cal S}$ is invariant under the operator ${\bf A}_{v}$ by Corollary \ref{ceig1}, hence
the space  ${\cal Z}_{v}$ is also invariant by (\ref{com}): ${\bf A}_{v}Z\in {\cal Z}_v$
for {\it sufficiently smooth}  $Z\in {\cal Z}_{v}$.
}
\end{remark}
In Section \ref{lin-dyn} below we  prove the following proposition.
\begin{pro}\la{lindecay}
 Let all conditions of Theorem \ref{main} hold, $|v_1|\le\ti v<1$, and $X_0\in{\cal Z}_{v_1}\cap\cE_\beta$ with $\beta>4$.
Then 
 $X(t)=e^{{\bf A}_1t}X_0\in C(\R, {\cal Z}_{v_1}\cap {\cal E}_{-\beta})$, and the following decay holds
\be\la{frozenest}
\Vert e^{{\bf A}_1t}X_0\Vert_{-\beta}\le
\frac{C_{\beta}(\ti v)}{(1+|t|)^2}\Vert X_0\Vert_{\beta},~~~~~~~~
\,\,\,t\in\R.
\ee
\end{pro}
\setcounter{equation}{0}
\section{Decay of transversal componentI}
\begin{pro}\la{pdec}
 Let all conditions of Theorem \ref{main} hold. Then $t_*=\infty$, and
\be\la{Zdec}
\Vert Z(t)\Vert_{-\beta}\le\frac {C(\rho,\ov v,d_\beta)}{(1+|t|)^2},\qquad t\ge0.
\ee
\end{pro}
\begin{proof}
{\it Step i)} Equation (\ref{redy1}) leads to
$$
{\bf P}_1Z^+(t)=e^{{\bf A}_1t}{\bf P}_1Z^+(0)+\int_0^te^{A_1(t-s)}{\bf P}_1\big({\bf B}(s)Z^+(s)+\ti N^+(s)\big)ds.
$$
We have used here that projection ${\bf P}_1:={\bf P}_{v(t_1)}$ commutes with the group $e^{{\bf A}_1t}$ since the space 
${\cal Z}_1:={\bf P}_1{\cal E}$ is invariant with respect to $e^{{\bf A}_1t}$ by Remark \ref{rZ}.
Applying (\ref{frozenest}), we obtain:
\be\la{bPZ}
\Vert {\bf P}_1Z^+(t)\Vert_{-\beta}
\le\frac{C\Vert {\bf P}_1Z^+(0)\Vert_{\beta}}{(1+|t|)^2}+C\int_0^t\frac{\Vert {\bf P}_1[{\bf B}(s)Z^+(s)+\ti N^+(s)]\Vert_{\beta}}{(1+|t-s|)^2}ds.
\ee
The operator ${\bf P}_1={\bf I}-{\bf\Pi}_1$ is continuous in ${\cal E}_\beta$ by (\ref{Piv}).
Hence, (\ref{B1Z1est}),  (\ref{N1est}), and (\ref{bPZ})  imply that 
\be\la{duhest}
\Vert {\bf P}_1Z^+(t)\Vert_{-\beta}
\le\frac{C\Vert Z(0)\Vert_{\beta}}{(1+|t|)^2}
+C\int_0^t\frac{\Vert Z(s)\Vert_{-\beta}}{(1+|t-s|)^2}\left[\int_s^{t_1}\Vert Z(\tau)\Vert_{-\beta}^2 d\tau+\Vert Z(s)\Vert_{-\beta}\right]ds
\ee
 for   $t_1<t_*'$  and $0\le t\le t_1$. Finally,  we are going to  replace ${\bf P}_1Z^+(t)$ by $Z(t)$ in the left hand side of (\ref{duhest}).
For  this we introduce the ``majorant'' 
\be\la{maj1}
M(t):=
\sup_{s\in[0,t]}(1+s)^2\Vert Z(s)\Vert_{-\beta}
\ee
and reduce further the exit time. Denote by $\ve$ a fixed positive number which we will specify below.
\begin{defin} $t_{*}''$ is the exit time
\be\la{t*''}
t_*''=\sup \{t\in[0,t_*'):\ M(s)\le \ve,~~0\le s\le t\}.
\ee
\end{defin}
\begin{lemma}\la{Z1P1Z1} 
For sufficiently small $\ve>0$ and  $t_1<t_*''$,
the following bounds hold:
\be\la{Z1P1est}
\Vert Z(t)\Vert_{-\beta}\le C\Vert {\bf P}_1Z^+(t)\Vert_{-\beta},\quad0\le t \le t_1,
\ee
where $C$ depends only on $\rho$ and $\ti v$.
\end{lemma}
The proof is based on the symplectic orthogonality \eqref{PZ} and is similar to that of Lemma 9.2 in \cite{IKV2012}.
Lemma \ref{Z1P1Z1} together with \eqref{duhest}  imply:
$$
\Vert Z(t)\Vert_{-\beta}
\le C\Big(\frac{\Vert Z(0)\Vert_{\beta}}{(1+t)^2}+\int\limits_0^t\frac{\Vert Z(s)\Vert_{-\beta}}{(1+t-\!s)^2}
\Big[\int\limits_s^{t_1}\Vert Z(\tau)\Vert_{-\beta}^2 d\tau+\Vert Z(s)\Vert_{-\beta}\Big]ds\Big),~~t_1<t_*''.
$$
Multiplying both sides  by $(1+t)^2$ and taking the supremum in $t\in[0,t_1]$, we get:
$$
M(t_1) \le C\Big(\Vert Z(0)\Vert_{\beta}+\!
\sup_{t\in[0,t_1]}\int\limits_0^t \frac{(1+t)^2}{(1+t-\!s)^2}\Big[\frac{M(s)}{(1+s)^{2}}
\int\limits_s^{t_1}\frac{M^2(\tau)d\tau}{(1+\tau)^{4}}+\frac{M^2(s)}{(1+s)^{4}}\Big]ds\Big),~~t_1<t_*''.
$$
Taking into account that $M(t)$ is a monotonically increasing function, we get:
\beqn\la{mest}
M(t_1)&\le& C\Vert Z(0)\Vert_{\beta}+C_1\big(M^3(t_1)+M^2(t_1)\big),~~~~t_1<t_*''.
\eeqn
This inequality implies that $M(t_1)$ is bounded for $t_1<t_*''$
and moreover
\be\la{m2est}
M(t_1)\le C_2\Vert Z(0)\Vert_{\beta},~~~~~~~~~t_1<t_*''\,,
\ee
since $M(0)$ is sufficiently small by (\ref{close}).
\\
{\it Step ii)} The constant $C_2$ in  (\ref{m2est}) does not depend on
$t_*$, $t_*'$, and $t_*''$ by Lemma \ref{Z1P1Z1}.
We choose $d_{\beta}$ in (\ref{close}) so small that
$\Vert Z(0)\Vert_{\beta}<\ve/(2C_2)$.
This is possible due to (\ref{closeZ}).
Then  (\ref{m2est}) implies that $t''_*=t'_*$ and therefore  (\ref{m2est}) holds for all $t_1<t'_*$.
By (\ref{vw}), 
$$
u(s)-v(t_{1})=u(s)-v(s)+v(s)-v(t_{1})=\dot c(s)+\ds\int_s^{t_1}\dot v(\tau)d\tau.
$$
Hence, Lemma \ref{mod}  and (\ref {dd1})   imply that
\beqn\nonumber
 |\ell_{1}(t)|&\le& \int_{t}^{t_{1}}\left( |\dot{c}(s)|+\int_{s}^{t_{1}}|\dot{v}(\tau )|d\tau\right)ds\\
\la{d1est} 
& \le& CM^2(t_1)\int_{t}^{t_{1}}\left(\frac{1}{(1+s)^4}+\int_{s}^{t_{1}}\frac{d\tau}{1+\tau)^4}\right)ds\le C_1M^2(t_1)\le C_2\ve^2,\quad t<t_*'.
\eeqn 
  We choose $\ve$ so
small that the right hand side in (\ref{d1est}) does not exceed one. Then $t'_*=t_*$.
Therefore, (\ref{m2est}) holds for all $t_1<t_*$.  Hence,
\be\la{Zt}
\Vert Z(t)\Vert_{-\beta}<z_{-\beta}(\ti v)/2,~~~~~0\le t < t_*.
\ee
if $\Vert Z(0)\Vert_{\beta}$ is sufficiently small.  Finally, this implies that $t_*=\infty$, hence also $t''_*=t'_*=\infty$ and
(\ref{m2est}) holds for all $t_1>0$ if $d_{\beta}$ is small enough. 
\end{proof}
\section{Soliton Asymptotics}\la{sol-as}
\setcounter{equation}{0}
Here we prove our main Theorem \ref{main}. 
First we will prove the asymptotics (\ref{qq-as}) for the vector components
and afterwards the asymptotics (\ref{AP-as}) for the fields.
\\
{\bf Asymptotics for the vector components}.
From (\ref{dot-q}),  (\ref{AA}), (\ref{reduced}), and  (\ref{redN})
it follows  that
\beqn\nonumber
\dot q(t)&=&\dot b(t)+\dot r=v(t)+\dot c(t)+\frac 1m(\pi(t)-\langle\Lam(t),\rho\rangle)-{\w}JPr+{\cal O}(\Vert Z(t)\Vert_{-\beta})\\
\la{dq}
&=&v(t)+\dot c(t)+{\cal O}(\Vert Z(t)\Vert_{-\beta}).
\eeqn
Further, (\ref{parameq}) and (\ref{Zdec}) imply that
\be\la{bv}
|\dot c(t)|+|\dot v(t)|\le \frac {C_1(\rho,\ti v,d_\beta)}{(1+t)^{4}},~~~~~~t\ge0.
\ee
Therefore,
\be\la{cv-as} 
c(t)=c_+ +{\cal O}(t^{-3}),\qquad v(t)=v_+ +{\cal O}(t^{-3}),\quad  t\to\infty.
\ee 
Hence, 
(\ref{Zdec}) and  (\ref{dq})--(\ref{bv}) imply that
\be\la{qbQ}
\dot q(t)=v_++{\cal O}(t^{-2}),\quad b(t)=c(t)+\int_0^tv(s)ds=v_+t+a_++{\cal O}(t^{-2})
\ee
Hence, asymptotics (\ref{qq-as}) for $q(t)$ follows by \eqref{add} and (\ref{Zdec}).  
Finally, \eqref{pM}, \eqref{add}, \eqref{Zdec} and \eqref{cv-as} imply that
\beqn\nonumber
\om(t)&\!\!\!=\!\!\!&\frac 1{I}\Big(M+\langle A(x,t),J\varrho(x\!-\!q(t))\rangle\Big)=\frac 1{I}\Big(M+\langle A_{v(t)}(y)+\Lam(y,t),J\varrho(y-r(t))\rangle\Big)\\
\la{ww}
&\!\!\!=\!\!\!&{\w}(t)+\frac 1I\Big(\langle A_{v(t)},J(\varrho(y\!-\!r(t))-\varrho(y))\rangle+\Lambda(y,t),J(\varrho(y\!-\!r(t))\rangle\Big)
={\w}(t)+ {\cal O}(t^{-2}),~~ t\to\infty. 
\eeqn
Here ${\w}(t)={\w}(v(t))$.
\\
{\bf Asymptotics for the fields}.
For the field part of the solution, $F(x,t)=(A(x, t), \Pi(x, t))$, we  define the accompanying soliton field as 
$F_{{\rm v}(t)}(t) = (A_{{\rm v}(t)}(x-q(t)), \Pi_{{\rm v}(t)}(t)(x-q(t)))$, where we now set 
${\rm v}(t) = \dot q (t)$, cf. \eqref{dq}. Then for the difference 
$\Phi(t)=(\Lam(t),\Psi(t)) = F(t) -F_{{\rm v}(t)}(t)$  the first two equations of  \eqref{mls3} imply 
the inhomogeneous  equation
\be\la{ZZ}
\dot \Phi(x,t)=\!\begin{pmatrix} 0& 1\\
\Delta&0
\end{pmatrix}\!\Phi(x,t)-\!R(t),\quad R(t)=\dot {\rm v}(t)\cdot\!\na_{{\rm v}}F_{{\rm v}(t)}(x\!-\!q(t))+(\w(t)\!-\!\om(t))\!\begin{pmatrix}
0\\ J\varrho(x\!-\!q(t))
\end{pmatrix}.
\ee
Hence,
\beqn\nonumber
\Phi(t)&=&W_0(t)\Phi(0)-\int_0^t W_0(t-s)R(s)ds\\
\nonumber
&=&W_0(t)\Big(\Phi(0)-\int_0^{\infty}W_0(-s)R(s)ds\Big)+\int_t^{\infty}W_0(t-s)R(s)ds=W_0(t)\Psi_++r_+(t).
\eeqn
The last two equations of \eqref{mls3}  imply
\beqn\nonumber
m\dot {\rm v}(t)&=&m\ddot q(t)=-\langle  A(x,t)\cdot \dot q(t),\nabla\rho(x\!-\!q(t))\rangle+\om(t) \langle \na\cdot JA(x,t),  \varrho(x\!-\!q(t))\rangle\\
\la{rmv}
&-&\langle \Pi(x,t),\rho(x\!-\!q(t))\rangle+\langle A(x,t),\dot q(t)\cdot\na\rho(x\!-\!q(t))\rangle\Big)
\eeqn
where $A(x,t)=A_{{\rm v}(t)}(x-q(t))+\Lam(x,t)$, $\Pi(x,t)=-({\rm v}(t)\cdot\na)A_{{\rm v}(t)}(x-q(t))+\Psi(x,t)$.
One has
\beqn\nonumber
&&-\langle  A_{{\rm v}(t)}\cdot {\rm v}(t),\nabla\rho\rangle+\om(t) \langle \na\cdot JA_{{\rm v}(t)},  \varrho\rangle
+\langle ({\rm v}(t)\cdot\na) A_{{\rm v}(t)},\rho)\rangle+\langle A_{{\rm v}(t)}, {\rm v}(t)\cdot\na\rho)\rangle\\
\nonumber
&&=-\langle  A_{{\rm v}(t)}\cdot {\rm v}(t),\nabla\rho\rangle+\om(t) \langle \na\cdot JA_{{\rm v}(t)},  \varrho\rangle
=-{\w}\langle  \frac{J\na\hat\rho\cdot {\rm v}(t)}{\hat D_0(k,{\rm v}(t)}),k\hat\rho\rangle+\om(t) \langle \frac{({\rm v}(t)\cdot Jk)\hat\rho}{\hat D_0(k,{\rm v}(t))}),  \na\hat\rho\rangle\\
\nonumber
&&=({\w}(t)-\om(t))P({\rm v}(t))J{\rm v}(t)={\cal O}(t^{-2}),~~ t\to\infty. 
\eeqn
by \eqref{ww}  and \eqref{qbQ}. Hence, $\dot {\rm v}(t)={\cal O}(t^{-2})$ by \eqref{Zdec}, \eqref{qbQ}  and \eqref{rmv}.
Now \eqref{ww} and \eqref{ZZ} imply that  $\Vert R(t)\Vert_{\cal F}  = {\cal O}(t^{-2})$.
Therefore $\Psi_{+} =\Phi(0)-\int_0^{\infty}W_0(-s)R(s)ds\in{\cal F}$ and  
$\Vert r_{+}(t)\Vert_{\cal F}= {\cal O}(t^{-1})$ due to the unitarity of $W_0(t)$.
\hfill $\Box$ 
\setcounter{equation}{0}
\section{Linearized dynamics}\la{lin-dyn}
\subsection{Solving the linearized equation }\la{lin-sol}
We change variables in equation \eqref{Avv} to simplify its structure. Put $\phi=\frac{1}{m}\big(\pi-\langle \Lam,\rho\rangle\big)$. 
If we prove a decay fot $\phi$ and $\Lam$, then $\pi$  has the corresponding decay as well.
By \eqref{AA} and \eqref{B-op},
\beqn\nonumber
&&\dot\phi=\frac{1}{m}\big(\dot\pi-\langle \dot\Lam,\rho\rangle\big)
=\frac{1}{m}\Big(B^*\Lam-Sr-\frac{\w}{m} PJ\pi-\langle \Psi+v\cdot\na\Lam,\rho\rangle\Big)\\
\nonumber
&&=\frac{1}{m}\Big(-\langle v\cdot \Lam,\na\rho\rangle+{\w}\langle \na\cdot J\Lam,\varrho\rangle-{\w} PJ\phi
+\nu PJv -S^{-}r-\langle \Psi+v\cdot\na\Lam,\rho\rangle\Big)
\eeqn
where 
\be\la{nu-def}
\nu=\frac{1}{I}\big(\langle \Lam, J\varrho\rangle-r\cdot PJv\big),\quad S^{-}r=(\frac {\w^2 }m P^2+Q+{\w}^2 F)r.
\ee
We denote $v=v_1$,  ${\cal X}=(\Lambda, \Psi, r, \phi)$.  Equation  \eqref{Avv} implies
\be\la{bfA}
\dot{\cal X}\!={\cal A}{\cal X}=\left(
\!\ba{l}
\Psi+v\cdot\na\Lambda \\
\De\Lambda+v\cdot\na\Psi+ {\cal P}[\rho\phi - \frac {\w}m\rho JPr-\!v(r\cdot\na\rho)]+\big({\w}(r\cdot \na)-\nu\big)J\varrho \\
\phi-\frac{\w}m JPr\\
-\frac{1}m\Big(\langle v\cdot \Lam,\na\rho\rangle-{\w}\langle \na\cdot \!J\Lam,\varrho\rangle+{\w} PJ \phi-\nu PJv+S^{-}r
+\langle \Psi+v\cdot\na\Lam,\rho\rangle\Big)
\ea\!\!\!\right)
\ee
Now we construct and study the resolvent of $\cal A$.  We apply the Fourier-Laplace transform
\be\la{FL}
\ti {\cal X}(\lam)=\int_0^\infty e^{-\lam t}{\cal X}(t)dt,~~~~~~~\rRe\lam>0
\ee
According to Proposition \ref{lindecay}, we expect that
the solution ${\cal X}(t)$ is bounded in the norm $\Vert\cdot\Vert_{-\beta}$.
Then the integral (\ref{FL}) converges and is analytic for $\rRe\lam>0$, and
\be\la{PW}
\Vert\ti {\cal X}(\lam)\Vert_{-\beta}\le \ds\frac{C}{\rRe\lam},~~~~~~~\rRe\lam>0.
\ee
Let us derive an equation for $\ti {\cal X}(\lam)$ which is equivalent to the
Cauchy problem for (\ref{bfA}) with the initial condition ${\cal X}(0)={\cal X}_0\in\cE_{-\beta}$.
Applying the Fourier-Laplace transform to (\ref{bfA}), we get
\be\la{FLAs}
\ti {\cal X}(\lam)=-({\cal A}-\lam)^{-1}{\cal X}_0,~~~~~~~~\rRe\lam>0.
\ee
Due to \eqref{bfA}, equation (\ref{FLAs}) reads
\be\la{eq1}
\left\{\!
\ba{l}
\ti\Psi+v\cdot\na\ti\Lambda -\lam\ti\Lambda=-\Lambda_0\\
\De\ti\Lambda+v\cdot\!\na\ti\Psi+{\cal P}\big[\big(\ti\phi - \frac{\w}m JP \ti r-v(\ti r\cdot\na)\big)\rho\big]
+\big( {\w}(\ti r\cdot\na)\!-\ti\nu)J\varrho-\lam\ti\Psi=-\Psi_0\\
\ti\phi-\frac{\w}m JP\ti r-\lam \ti r=-r_0\\
\langle v\cdot \ti\Lam,\na\rho\rangle+{\w} PJ\ti\phi-{\w}\langle \na\cdot \!J\ti\Lam,\varrho\rangle+S^{-}\ti r-\ti\nu PJv
+\langle \ti\Psi+v\cdot\na\ti\Lam,\rho\rangle+m\lam\ti\phi=m\phi_0,
\ea
\right.
\ee
We rewrite this system as
\be\la{eq-main}
\left\{\!
\ba{l}
\ti\Psi+v\cdot\na\ti\Lambda -\lam\ti\Lambda=-\Lambda_0\\
\De\ti\Lambda+v\cdot\!\na\ti\Psi+{\cal P}\big[\big(\lam\ti r-v(\ti r\cdot\na)\big)\rho\big]
+\big( {\w}(\ti r\cdot\na)\!-\ti\nu)J\varrho-\lam\ti\Psi=-\Psi_0+{\cal P}[r_0\rho]\\
\ti\phi-\frac{\w}m P\ti r-\lam \ti r=-r_0\\
\langle v\cdot \ti\Lam,\na\rho\rangle+({\w} PJ+m\lam)(\lam+\frac{\w}m JP)\ti r-{\w}\langle \na\cdot \!J\ti\Lam,\varrho\rangle
+S^{-}\ti r-\ti\nu PJv+\lam\langle \ti\Lambda,\rho\rangle\\
=\pi_0+({\w} PJ\!+m\lam)r_0
\ea\right.
\ee
{\it Step i)} Let us study the first two equations. In Fourier space
 they become
$$
\left.
\ba{rcl}
\hat{\ti\Psi}(k)-i(vk)\hat{\ti\Lam}(k)-\lam\hat{\ti \Lam}(k)&=&-\hat\Lam_0(k)
\\
-k^2\hat{\ti \Lam}(k)-i(vk)\hat{\ti\Psi}(k)-\lam\hat{\ti \Psi}(k)&=&-\hat\Psi_0(k)+\widehat{{\cal P}r_0\rho}(k)-K(k)
\ea\right|,~~~(vk)=v\cdot k,~~k\in\R^2\5.
$$
where
\be\la{hat-Pi}
K(k)=\lam\widehat{{\cal P}\ti r\rho}(k)+i(k\ti r)\widehat{{\cal P} v\rho}(k)+\big(i\ti \nu-{\w} (k \ti r)\big)J\na\hat\rho,
\quad (k\ti r)=k\cdot\ti r.
\ee 
Let us invert the matrix of the system and obtain
$$
\left(
\ba{cc}
-(i(vk)+\lam) & 1 \\
-k^2 & -(i(vk)+\lam)
\ea
\right)^{-1}=\frac{1}{\hat D(\lam)}\left(
\ba{cc}
-(i(vk)+\lam) & -1 \\
k^2 & -(i(vk)+\lam)
\ea\right),
$$
where $\hat D(\lam):=(i(vk)+\lam)^2+k^2$ does not vanish since $\rRe\lam>0$ and $|v|<1$. 
Hence,
\be\la{hat-Lam}
\hat{\ti \Lam}=\frac{(i(vk)+\lam)\hat\Lam_0+\hat\Psi_0-\widehat{{\cal P}r_0\rho}+K}{\hat D(\lam)}=\frac{K_0}{\hat D(\lam)}+\frac{K}{\hat D(\lam)},
\ee
where
\be\la{K0}
K_0:=(i(vk)+\lam)\hat\Lam_0+\hat\Psi_0-\widehat{{\cal P}r_0\rho}.
\ee
Substituting \eqref{hat-Lam} into the forth equation of \eqref{eq-main} , we get
\be\la{eq-Pi}
i\langle \frac{v\cdot K}{\hat D(\lam)},k\hat\rho\rangle-{\w}\langle \frac{k\cdot J K}{\hat D(\lam)},\na\hat\rho\rangle
+\big(Q+{\w}^2 F+{\w}\lam(PJ+JP)+m\lam^2\big)\ti r-\ti\nu PJv+\lam\langle \frac{K}{\hat D(\lam)},\hat\rho\rangle=G_0(\lam),
\ee
where
\be\la{F1}
G_0(\lam)=\pi_0+({\w} PJ+m\lam)r_0-i\langle\frac{ v\cdot K_{0}}{\hat D(\lam)},k\hat\rho\rangle
+{\w}\langle\frac{ k\cdot JK_0}{\hat D(\lam)},\na\hat\rho\rangle-\lam\langle \frac{K_0}{\hat D(\lam)},\hat\rho\rangle.
\ee
\begin{defin}\la{def-p-lam} 
We will denote  by $P_{jl}(\lam)$, $Q_{jl}(\lam)$, $F_{jl}(\lam)$, etc. the function defined in \eqref{c-ro}, \eqref{h-ro} 
and \eqref{f-ro} with $\hat D(\lam)$ instead of $\hat D_0$.
\end{defin}
Due to \eqref{Pi-e} and \eqref{hat-Pi},
\beqn\nonumber
&&\langle\frac {K}{\hat D(\lam)},\hat\rho\rangle=\lam\langle \frac{\widehat{{\cal P}\ti r\rho}}{\hat D(\lam)},\hat\rho\rangle
-{\w}\langle \frac{(k {\ti r})J\na\hat\rho}{\hat D(\lam)},\hat\rho\rangle
+i\langle\frac{(k\ti r)\widehat{{\cal P} v\rho}}{\hat D(\lam)},\hat\rho\rangle+i\ti \nu\langle\frac{ J\na\hat\rho}{\hat D(\lam)},\hat\rho\rangle\\
\nonumber
&&=\lam\int\frac{|\hat\rho|^2(\ti r\cdot Jk)}{k^2\hat D(\lam)}\begin{bmatrix} k_2\\-k_1\end{bmatrix} dk
-{\w}\!\int\frac{\hat\rho(k{\ti r})}{\hat D(\lam)}\begin{bmatrix} \na_2\hat\rho\\ -\na_1\hat\rho\end{bmatrix}dk
+i\!\int\frac{|\hat\rho|^2(k\ti r)(v\cdot Jk)}{k^2\hat D(\lam)}\begin{bmatrix} k_2\\-k_1\end{bmatrix} dk +i\ti \nu f(\lam)\\
\la{Pi-1}
&&=\big(\lam U(\lam)-{\w} JP(\lam) +i JS(\lam)\big)\ti r+i\ti \nu f.
\eeqn
Here
\beqn\la{Sj}
&& f(\lam)=\langle\frac{J\na\hat\rho}{\hat D(\lam)},\hat\rho\rangle,\qquad 
S_{jl}(\lam)=\int\frac{ (v\cdot Jk)|\hat\rho|^2k_j k_l}{k^2\hat D(\lam)} dk,~~j,l=1,2,\\
\la{tau-p}
&&U_{jj}(\lam)=\int\frac{(k^2-k_j^2)|\hat\rho|^2dk}{k^2\hat D(\lam)},
\quad U_{jl}(\lam)=-\int\frac{k_jk_l|\hat\rho|^2dk}{k^2\hat D(\lam)},~~j\ne l. 
 \eeqn
Further, 
\beqn\nonumber
\langle\frac{ k\cdot JK}{\hat D(\lam)},\na\hat\rho\rangle
&\!\!\!=\!\!\!&-\lam\int\frac{(\ti r\cdot Jk)\hat\rho}{\hat D(\lam)}\na\hat\rho dk+{\w}\int\frac{(k\ti r)(k\cdot\na\hat\rho)}{\hat D(\lam)}\na\hat\rho dk
-i\!\int\! \frac{(k\ti r) (v\cdot Jk)\hat\rho}{\hat D(\lam)}\na\hat\rho dk-i\ti \nu g(\lam)\\
\label{Pi-5}
&\!\!\!=\!\!\!&\big(\lam P(\lam)J+{\w} F(\lam)-iR(\lam)\big)\ti r-i\ti \nu h(\lam),
\eeqn
where 
\be\la{fj-def}
h(\lam):=\langle \frac{k\cdot\na\hat\rho}{\hat D(\lam)},\na\hat\rho\rangle=\int\frac{k|\na\hat\rho|^2}{\hat D(\lam)} dk,
\quad R_{jl}(\lam)=\int \frac{(v\cdot Jk)\hat\rho k_j\na_l\hat\rho}{\hat D(\lam)} dk.
\ee
Moreover, \eqref{Pi-e} and \eqref{hat-Pi}   imply
\beqn\nonumber
&&\langle \frac {v\cdot K}{\hat D(\lam)},k\hat\rho\rangle
=i\langle\frac{v^2k^2-(vk)^2}{k^2\hat D(\lam)}\hat\rho(k\ti r), k\hat\rho\rangle
+i\ti\nu\langle \frac{v\cdot J\na\hat\rho}{\hat D(\lam)},k\hat\rho\rangle-{\w}\langle\frac{(k\ti r)(v\cdot J\nabla\hat\rho)}{\hat D(\lam)},k\hat\rho\rangle\\
\label{Pi-3}
&&+\lambda\langle\frac{(v\cdot Jk)(r\cdot Jk)\hat\rho}{k^2\hat D(\lam)},k\hat\rho\rangle=\big(iQ(\lam)-{\w} R(\lam)-\lambda S(\lam)J\big)\ti r-i\ti\nu P(\lam)Jv.
\eeqn
since 
$$
\int\frac{k_jk_l\hat\rho(v\cdot J\na\hat\rho)}{\hat D(\lam)} dk=\int \frac{(v\cdot Jk)\hat\rho k_j\na_l\hat\rho}{\hat D(\lam)} dk.
$$
Substituting \eqref{Pi-1},  \eqref{Pi-5} and  \eqref{Pi-3}  into the LHS of \eqref{eq-Pi}, we obtain
\be\la{M12}
\Big(\breve Q+{\w}^2\breve F+\lam^2 U(\lam)+{\w}\lam\big(\breve PJ\!+\!J\breve P\big)+i\lam(JS(\lam)-S(\lam)J)+m\lam^2\Big)\ti r
+\ti\nu\big(i\lam f(\lam)+i{\w} h(\lam)-\breve PJv\big)=G_0(\lam),
\ee
where  $\breve Q=\breve Q(\lam):=Q-Q(\lam)$, etc.  
Let us  express $\ti\nu$ in terms of $\ti r$. By  \eqref{Pi-e} and  \eqref{hat-Pi}, 
\beqn\nonumber
\!\!\!\!\!&&\!\!\!\!\!\langle \frac {K}{\hat D(\lam)}, J\hat\varrho\rangle
=-\langle \frac {(k\ti r)\hat\rho}{\hat D(\lam)}v, J\na\hat\rho\rangle
+i\lambda\langle \frac {\hat\rho}{\hat D(\lam)}\ti r, J\na\hat\rho\rangle
-i{\w}\langle\frac{(k\ti r)\na\hat\rho}{\hat D(\lam)},\na\hat\rho\rangle-\ti \nu\langle\frac{\na\hat\rho}{\hat D(\lam)},\na\hat\rho\rangle\\
\nonumber
\!\!\!\!\!&&\!\!\!\!\!=\int\frac {\hat\rho(k\ti r)(v_2\na_1\hat\rho-v_1\na_2\hat\rho)}{\hat D(\lam)} dk+i\lam\ti r\cdot f(\lam)
-i{\w}\ti r\cdot\int\frac{k|\na\hat\rho|^2}{\hat D(\lam)} dk-\ti \nu\varkappa(\lam)\\
\label{Pi-4}
&&=\ti r\cdot \big(P(\lam)Jv+i\lam f(\lam)-i{\w} h(\lam)\big)-\ti\nu\varkappa(\lam),
\eeqn
where we denote  (cf. \eqref{kappa}) 
$\varkappa(\lam)=\int\frac{|\na\hat\rho|^2}{\hat D(\lam)}dk$.
 Hence \eqref{nu-def}  and \eqref{hat-Lam} imply
\beqn\nonumber
I\ti\nu&=&\ti r\cdot \big(P(\lam)Jv+i\lam f(\lam)-i{\w} h(\lam)\big)-\ti\nu\varkappa(\lam)+\langle \frac{K_0}{\hat D(\lam)},J\hat\varrho\rangle-\ti r\cdot PJv\\
\nonumber
&=&\ti r\cdot \big(i\lam f(\lam)-i{\w} h(\lam)-\breve P(\lam)Jv\big)-\ti\nu\varkappa(\lam)+\langle \frac{K_0}{\hat D(\lam)},J\hat\varrho\rangle.
\eeqn
Therefore,
\be\la{nu-rep}
\ti\nu=\frac{1}{\kappa(\lam)}\Big(\ti r\cdot \big(i\lam f(\lam)-i{\w} h(\lam)-\breve P(\lam)Jv\big)+\langle \frac{K_0}{\hat D(\lam)},J\hat\varrho\rangle\Big),
\quad \kappa(\lam):=I+\varkappa(\lam).
\ee
Evidently, $\kappa(0)>0$. Moreover, $\kappa(\lam)\ne 0$ for $\rRe\lam>0$. In Appendix  we will show that $\kappa(i\mu+0)\ne 0$ for $\mu\in\R$
 in the case  $v\ne 0$. 

Substituting \eqref{nu-rep} into \eqref{M12}, we get
\beqn \nonumber
&&\Big(\breve Q+{\w}^2\breve F+\lam^2 U(\lam)+{\w}\lam\big(\breve PJ\!+\!J\breve P\big)+i\lam(JS(\lam)-S(\lam)J)+m\lam^2\Big)\ti r\\
\la{M13}
&&+r\cdot\frac{i\lam f(\lam)-i{\w} h(\lam)-\breve P(\lam)Jv}{\kappa(\lam)}\Big(i\lam f(\lam)+i{\w} h(\lam)-\breve P(\lam)Jv\Big)=G_0^+(\lam),
\eeqn
where
\be\la{F1+}
G^+_0(\lam):=G_0(\lam)-\frac{1}{\kappa(\lam)}\langle \frac{\hat K_0}{\hat D(\lam},J\hat\varrho\rangle\Big(i\lam f(\lam)+i{\w} h(\lam)-\breve PJv\Big)
\ee
We rewrite \eqref{M13}  as the following system
\be\la{M-def}
L(\lam)\begin{pmatrix} \ti r_1\\\ \ti r_2\end{pmatrix}=G_0^+(\lam),
\quad L_{jl}(\lam)=a_{jl}(\lam)+b_{jl}(\lam)
\ee
Here, 
\beqn\nonumber
a_{jl}&=&\breve Q_{jl}+{\w}^2\breve F_{jl}+\lam^2(m\delta_{jl}+U_{jl})+{\w}\lam\big(\breve PJ\!+\!J\breve P\big)_{jl}+i\lam(JS(\lam)-S(\lam)J)_{jl}\\
\label{aabb}
b_{jl}&=&\frac{1}{\kappa}\big(i\lam f_j-i{\w} h_j-(\breve PJv)_j\big)\big(i\lam f_l+i{\w} h_l-(\breve PJv)_l\big)
\eeqn
Note that 
\be\la{PPJ}
 \breve PJ+J\breve P
=\begin{pmatrix} 0 & \breve P_{11} +\breve P_{22} \\ -\breve P_{11} -\breve P_{22}& 0 \end{pmatrix},\quad
 JS-SJ=\begin{pmatrix} 2S_{12} & S_{22} -S_{11} \\ S_{22} -S_{11} & -2S_{12} \end{pmatrix}
\ee
Moreover.
$P_{12}(\lam)=F_{12}(\lam)=U_{12}(\lam)=S_{11}(\lam)=S_{22}(\lam)=0$, $f_2(\lam)=h_2(\lam)=0$.
Hence
\beqn\nonumber
a_{jj}&=&\breve Q_{jj}+{\w}^2\breve F_{jj}+\lam^2(m+U_{jj})+2i(-1)^{j+1}S_{12},\quad j=1,2,\\
\la{aa}
a_{12}&=&{\w}\lam(  \breve P_{11} +\breve P_{22}),\quad a_{21}=-{\w}\lam(\breve P_{11} +\breve P_{22})\\
\la{bb}
b_{11}&=&\frac{1}{\kappa}\big( {\w}^2 h_1^2-\lam^2 f_1^2\big),\quad  b_{22}=\frac{1}{\kappa}P_{22}^2v^2,\quad b_{12}=b_{21}=0
\eeqn
The invertibility of the matrix $L(\lam)$ for $\rRe\lam>0$ follows from the next lemma.
\begin{lemma}\la{KK} 
The operator ${\cal A}-\lam:{\cal E}\to{\cal E}$ has a bounded inverse operator for $\rRe\lam>0$.
\end{lemma}
The proof is similar to that of  Lemma 7.1 in \cite{K2025}. By \eqref{M-def},
 \beqn\nonumber
L^{-1}(\lam)&=&\frac{1}{{\rm det}\, L}\begin{pmatrix} 
a_{22} +b_{22}&-(a_{12}+b_{12})\\
-(a_{21}+b_{21})&  a_{11} +b_{11}
\end{pmatrix},\quad  {\rm where}\\\\
\nonumber
\la{detM}
 {\rm det}\, L&=&(a_{11} +b_{11})(a_{22} +b_{22})-(a_{12}+b_{12})(a_{21}+c_{21}),
\eeqn
\subsection{Asymptotics of matrix $L^{-1}(\lam)$}\la{det-as}
Now we assume that $v=(|v|,0)$.  
\begin{lemma}\la{Dr-as}
{\it i)} Let $\rho$ satisfy  \eqref{rosym} and   \eqref{zero1} 
and $f\in L^2_{\beta}$ with $\beta>4$. Then
\be\la{int-rho-f-as}
\langle\frac{\hat\rho}{\hat D(\lam)}, \hat f\rangle=\langle\frac{\hat\rho}{\hat D_0}, \hat f\rangle
-\lam \langle\frac{2i|v|k_1\hat\rho}{\hat D_0^2}, \hat f\rangle
-\lam^2\langle\frac{(k^2+3v^2k_1^2)\hat\rho}{\hat D_0^3}, \hat f\rangle
+{\cal O}(|\lam|^3),\quad \lam\to 0,\quad \rRe\lam> 0.
\ee
Moreover,
\beqn\nonumber
\langle\frac{k\hat\rho}{\hat D(\lam)}, \hat f\rangle&=&\langle\frac{k\hat\rho}{\hat D_0}, \hat f\rangle
-\lam \langle\frac{2i|v|k_1k\hat\rho}{\hat D_0^2}, \hat f\rangle
-\lam^2\langle\frac{(k^2+3v^2k_1^2)k\hat\rho}{\hat D_0^3}, \hat f\rangle\\
\la{int-rho-f-as+}
&+&4\lam^3\langle\frac{i|v|k_1(v^2k_1^2+k^2)k\hat\rho}{\hat D_0^4}, \hat f\rangle+{\cal O}(|\lam|^4),\quad \lam\to 0,\quad \rRe\lam> 0.
\eeqn
\\
{\it iii)} 
Let  $\rho_j$  satisfy  \eqref{rosym}   with $|\al|\le 3$ and  \eqref{zero1}. Then for $\lam\to 0$, $\rRe\lam> 0$,
\be\la{int-rho-as}
\langle\frac{\hat\rho_1}{\hat D(\lam)}, \hat\rho_2\rangle=\langle\frac{\hat\rho_1}{\hat D_0}, \hat\rho_2\rangle
-\lam^2\langle\frac{(k^2+3v^2k_1^2)\hat\rho_1}{\hat D_0^3}, \hat\rho_2\rangle
+\lam^4\langle\frac{(v^4k_1^4+10v^2k^2k_1^2+k^4)\hat\rho_1}{\hat D_0^5}, \hat\rho_2\rangle+{\cal O}(|\lam|^6).
\ee
These asymptotics can be differentiated  appropriate number of times.
\end{lemma}
Formally, the expansions one can obtain by differentiation of $\hat D^{-1}(\lam)$. The rigorous  proof is given in \cite{KK2023}.
\begin{lemma}
Let conditions \eqref{rosym} and \eqref{zero1} hold. Then
\be\la{M-as}
L^{-1}(\lam)
=\begin{pmatrix} \ell_{11}\lam^{-2}+\ell_{11}'\lam^{-1}+{\cal O}(|\lam|) & \ell_{12}\lam^{-1}+\ell_{12}'\lam+ {\cal O}(|\lam|)\\
-\ell_{12}\lam^{-1}-\ell_{12}'\lam+ {\cal O}(|\lam|)& \ell_{22}\lam^{-2}+\ell_{22}''\lam^{-1}+{\cal O}(|\lam|)
\end{pmatrix},~~\lam\to 0,~~\rRe\lam>0
\ee
with $\ell_{jj}>0$, $j=1.2$. The asymptotics  can be differentiated two times. 
\end{lemma}
\begin{proof}
Note that
$$
P_{jj}^{(k)}(0)=Q_{jj}^{(k)}(0)=F_{jj}^{(k)}(0)=U_{jj}^{(k)}(0)=0, \quad j=1,2,\quad k=1,3
$$
by \eqref{c-ro}, \eqref{h-ro}, \eqref{f-ro},  \eqref{tau-p} . Further, 
$$
S_{jl}^{(k)}(0)=0,\quad k=0,2,4
$$
by  \eqref{Sj}. Hence asymptotics \eqref{int-rho-f-as+} imply
\beqn\nonumber
a_{jj}(\lam)&=&\alpha_{jj}\lam^2+\beta_{jj}\lam^4+{\cal O}(|\lam|^6),\\
\la{a-j}
a_{12}(\lam)&=&a_{21}(\lam)=\alpha_{12}\lam^3+\beta_{12}\lam^5+{\cal O}(|\lam|^7),\quad \lam\to 0,\quad \rRe\lam>0.
\eeqn
where
\beqn\nonumber
\alpha_{jj}&=&-\frac 12Q_{jj}''(0)-\frac{\w^2}2F_{jj}''(0)+m+U_{jj}(0)+2i(-1)^{j+1}S_{12}'(0)\\
\nonumber
&=&v^2\int\frac{k_j^2k_2^2(k^2+3v^2k_1^2)|\hat\rho|^2}{k^2\hat D_0^3}dk
+{\w}^2\int\frac{(k\cdot\na\hat\rho)k_j\na_j\hat\rho(k^2+3v^2k_1^2)}{\hat D_0^3}dk+m\\
\la{alij}
&+&\int \frac{(k^2-k_j^2)|\hat\rho|^2}{k^2\hat D_0}dk+4(-1)^{j+1}v^2\int\frac{k_1^2k_2^2|\hat\rho|^2}{k^2\hat D^2_0}dk>0
\eeqn
One has $\alpha_{11}>0$. Moreover,
\beqn\nonumber
\alpha_{22}&=&v^2\int\frac{k_2^4(k^2+3v^2k_1^2)|\hat\rho|^2}{k^2\hat D_0^3}dk
+{\w}^2\int\frac{k_2^2|\na\hat\rho|^2(k^2+3v^2k_1^2)}{\hat D_0^3}dk+m\\
\la{al22}
&+&\int \frac{k_1^2|\hat\rho|^2}{k^2\hat D_0}dk-4v^2\int\frac{k_1^2k_2^2|\hat\rho|^2}{k^2\hat D^2_0}dk>0
\eeqn
since
\beqn\nonumber
4\frac{v^2k_1^2k_2^2|\hat\rho|^2}{k^2\hat D_0^2}
&=& \frac{k_1^2|\hat\rho|^2}{k^2\hat D_0} \Big[ 2 \frac{|v|}2\cdot \frac{|v|k_2^2}{\hat D_0}+3\big(2\frac{1}{2}\cdot\frac{v^2k_2^2}{\hat D_0}\big)\Big]
\le  \frac{k_1^2|\hat\rho|^2}{k^2\hat D_0} \Big[\frac{v^2}4+\frac{v^2k_2^4}{\hat D_0^2}+3(\frac 14+\frac{|v|^4k_2^4}{\hat D_0^2})\Big]\\
\nonumber
&\le&  \frac{k_1^2|\hat\rho|^2}{k^2\hat D_0} \Big[1+\frac{v^2k_2^4(1+3v^2)}{\hat D_0^2}\Big]
\le  \frac{k_1^2|\hat\rho|^2}{k^2\hat D_0}+\frac{v^2k_2^4(k^2+3v^2k_1^2)|\hat\rho|^2}{k^2\hat D_0^3}
\eeqn
Now consider $b_{jl}$. Formulas \eqref{Sj}, \eqref{fj-def} and  \eqref{bb}  imply
\be\la{bb-ex}
b_{11}=\gamma_{11}\lam^2+\mu_{11}\lam^3+\nu_{11}\lam^4+{\cal O}(|\lam|^5), 
\quad  b_{22}=\mu_{jl}\lam^3+\nu_{11}\lam^4+{\cal O}(|\lam|^5),
\ee
where
$$
\gamma_{11}=\frac{{\w}^2}{\kappa(0)}(h_1'(0))^2=
-4\frac{{\w}^2v^2}{I+\varkappa(0)}\Big(\int\frac{k_1^2|\na\hat\rho|^2}{\hat D_0^2}dk\Big)^2
$$
One has
\beqn\nonumber
|\gamma_{11}|
&\le& \frac{4{\w}^2v^2}{I+\varkappa(0)} \int\frac{k_1^4|\na\hat\rho|^2}{\hat D_0^3}dk\cdot\int\frac{|\na\hat\rho|^2}{\hat D_0}dk 
\le 4{\w}^2v^2\int\frac{k_1^4|\na\hat\rho|^2}{\hat D_0^3}dk\\
\la{al-b}
&\le& {\w}^2\int\frac{k_1^2|\na\hat\rho|^2)(k^2+3v^2k_1^2)}{\hat D_0^3}dk<\alpha_{11}
\eeqn
Now \eqref{detM} together with \eqref{a-j}, \eqref{bb-ex} and \eqref{al-b}  imply that
\beqn\nonumber
{\rm det}\, L&=&(\al_{11}+\gamma_{11})\al_{22}\lam^4+\ell_5\lam^5+\ell_6\lam^6+{\cal O}(|\lam|^7)\\
\la{m-as}
\frac{1}{{\rm det}\, L}&=&\frac{1}{(\al_{11}+\gamma_{11})\al_{22}}\lam^{-4}+q_3\lam^{-3}+q_2\lam^{-2}+{\cal O}(\lam^{-1})
\eeqn
Finally, \eqref{M-as}, with $\ell_{11}=a_{22}^{-1}$ and  $\ell_{22}=(a_{11}+\gamma_{11})^{-1}$,  follows by  \eqref{detM} and \eqref{m-as}. 
\end{proof}
To obtain asymptotics for large $\lam$, we return to the coordinate representation. 
One has
$$
\hat D(\lam,k)=k^2+(i|v|k_1+\lam)^2=\frac{1}{\ga^2}(k_1+i|v|\lam\gamma^2)^2+k_2^2+\ga^2\lam^2,\qquad \ga:=1/\sqrt{1-v^2} 
$$
Hence,
\beqn\nonumber
{\cal R}(\lam,y):&\!\!=\!\!&F^{-1}_{k\to y}\ds\frac{1}{\hat D(\lam,k)}=\frac{1}{2\pi}\int \frac{e^{-iky}dk}{\hat D(\lam,k)}
=\frac{e^{-|v|\lam\gamma^2 y_1}}{2\pi}\int \frac{e^{-iky}dk}{\frac{1}{\ga^2}k_1^2+k_2^2+\ga^2\lam^2}\\
\la{g-rep0}
&\!\!=\!\!&
\gamma e^{-|v|\lam\gamma \tilde y_1}\int \frac{e^{-ik\tilde y}dk}{k_1^2+k_2^2+\ga^2\lam^2},~~ \tilde y=(\gamma y_1,y_2).
\eeqn
Let ${\cal R}(\lam)$ be  an integral operator with the integral kernel ${\cal R}(\lam,x-y)$. By \eqref{g-rep0}, 
\be\la{g-rep}
{\cal R}(\lam,x-y)=\gamma e^{-|v|\lam\gamma (\tilde x_1-\tilde y_1)}R(-\gamma^2\lam^2,(\tilde x-\tilde y)),
\ee
where  $R(\zeta, x-y)$ is the integral kernel of the resolvent $R(\zeta)=(-\Delta-\zeta)^{-1}$ of the 2D Laplacian. 
Recall that for $R(\zeta)$ the following asymptotics hold  \cite{A, KK2012}: 
 For  $s,k=0,1,2,...$  and 
$|l|\le 2$,  
\be\la{A}
 \Vert R^{(k)}(\zeta)\Vert_{H^s_\si\to H^{s+l}_{-\si}}
 ={\cal O}(|\zeta|^{-\frac{1+k-l}2}), \quad |\zeta|\to\infty,\quad \zeta\in\C\setminus [0,\infty),\quad \si>1/2+k.
 \ee
The asymptotics  and formula \eqref{g-rep} imply
\be\la{g-as}
 \Vert {\cal R}^{(k)}(\lam)\Vert_{H^s_\si \to H^{s+l}_{-\si}}
 ={\cal O}(|\lam|^{-1+l}), \quad |\lam|\to\infty,\quad\rRe\lam>0,\quad \si>1/2+k,\quad |l|\le 2.
 \ee
 \begin{lemma}\la{Lem-q-est}
Let conditions \eqref{rosym} holds. Then 
\be\la{M-as-}
L^{-1}(\lam)=\begin{bmatrix} 
{\cal O}(|\lam|^{-2})&{\cal O}(|\lam|^{-5})\\
{\cal O}(|\lam|^{-5})& {\cal O}(|\lam|^{-2})\\
\end{bmatrix},\quad \lam\to\infty,\quad \rRe\lam>0.
\ee
The asymptotics can be differentiated two times.
 \end{lemma}
 \begin{proof}
 Applying \eqref{g-as} with $s=1$, $l=-1$, $k=0,1,2$ and $\si>1/2+k$, we obtain 
 \beqn\nonumber
|\varkappa^{(k)}(\lam)|&=&|\langle {\cal R}^{(k)}(\lam)\varrho,\varrho\rangle|
\le C\Vert {\cal R}^{(k)}(i\mu+0)\varrho\Vert_{H^{-1}_{-\si}}\Vert\varrho\Vert_{H^1_{\si}}\\
\la{q-as}
&\le& C_1|\lam|^{-3}\Vert\varrho\Vert_{H^1_{\si}}^2\le C(\rho) |\lam|^{-3},\quad \lam\to\infty,\quad \rRe\lam>0..
\eeqn
For $Q_{jj}^{(k)}(\lam)$, $F_{jj}^{(k)}(\lam)$,   $P_{jj}^{(k)}(\lam)$,
 $U_{jj}^{(k)}(\lam)$, $S_{jl}^{(k)}(\lam)$, $f_1(\lam)$, $h_1(\lam)$ the similar bounds holds.  
 Therefore, \eqref{aa}   implies
 \beqn\nonumber
&& a_{jj}^{(k)}(\lam)+b_{jj}^{(k)}(i\lam)= m\lam^{2-k}+{\cal O}(|\lam|^{-1}),\\
\nonumber
&&a_{jl}^{(k)}(\lam)+b_{jl}^{(k)}(i\lam)={\cal O}(|\lam|^{-1}),\quad j\ne l,\quad \lam\to\infty,\quad \rRe\lam>0,\quad k=0,1,2.
\eeqn 
Hence,  \eqref{M-as-} follows. 
\end{proof}
\subsection{Asymptotics of $G_0^+(\lam)$}\la{pr-as}
\begin{lemma}\la{F-as}
Let $\rho$ satisfies  \eqref{rosym} and \eqref{zero1}, and $(\Lam_0,\Psi_0)\in {\cal E}_{\sigma}$ with $\beta>4$. Then
\be\la{F12-as}
G_0^+(\lam)=g_{2}\lam^2+g_{3}\lam^3+{\cal O}(\lam^4),\quad \lam\to 0,\quad \rRe\lam>0.
\ee
\end{lemma}
The asymptotics can be differentiated two times.
\begin{proof}
By  and  \eqref{F1+}, 
$$
G_0^+(\lam)=G_0(\lam)-\frac{1}{\kappa(\lam)}\langle \frac{\hat K_0}{\hat D(\lam)},J\hat\varrho\rangle V(\lam),\quad
V(\lam)=i\lam f(\lam)+i{\w} h(\lam)-\breve PJv
$$
where
\be\la{V-as}
V(\lam)={\rm v}_0 +{\rm v}_1\lam +{\rm v}_2\lam^2+{\rm v}_3\lam^3+{\cal O}(\lam^4),\quad \lam\to 0,\quad \rRe\lam>0.
\ee
by definitions   \eqref{c-ro}, \eqref{Sj}, \eqref{fj-def}   asymptotics \eqref{int-rho-f-as+}.
Moreover, in the case $v=(|v|,0)$, \eqref{K0} and \eqref{F1} imply
\beqn\nonumber
G_0(\lam)&=&\pi_0+({\w} PJ+m\lam)r_0-i|v|\langle \frac{\cdot K_{01}}{\hat D(\lam)},k\hat\rho\rangle
+{\w}\langle \frac{k\cdot J\hat K_0}{\hat D(\lam)},\na\hat\rho\rangle-\lam\langle \frac{\hat K_0}{\hat D(\lam)},\hat\rho\rangle\\
\nonumber
&=&\pi_0+({\w} PJ+m\lam)r_0\\
\nonumber
&-&i|v|\Big(i|v|\langle \frac{k\hat\rho}{\hat D(\lam)},k_1\ov{\hat\Lam}_{01}\rangle+\lam\langle \frac{\hat\rho}{\hat D(\lam)},k\ov {\hat\Lam}_{01}\rangle
 +\langle  \frac{k\hat\rho}{\hat D(\lam)},\ov{\hat\Psi}_{01}\rangle-\langle \frac{k\hat\rho}{\hat D(\lam)}, (\widehat{{\cal P}r_0\rho})_1\rangle\Big)\\
\nonumber
&+&{\w}\Big(i|v|\langle \frac{k_1\na\hat\rho}{\hat D(\lam)},k\cdot\! J\ov{\hat\Lam}_{0}\rangle
+\lam\langle \frac{\na\hat\rho}{\hat D(\lam)},k\cdot\! J\ov{\hat\Lam}_{0}\rangle
+\langle  \frac{\na\hat\rho}{\hat D(\lam)},k\cdot\! J\ov{\hat\Psi}_{0}\rangle+\langle \frac{\na\hat\rho}{\hat D(\lam)},(r_0\cdot\! Jk)\hat\rho\rangle\Big)\\
\la{F-split}
&+&\lam\Big(i|v|\langle \frac{k_1\hat\rho}{\hat D(\lam)},\ov{\hat\Lam}_{0}\rangle+\lam\langle \frac{\hat\rho}{\hat D(\lam)},\ov{\hat\Lam}_{0}\rangle
+\langle  \frac{\hat\rho}{\hat D(\lam)}, \ov{\hat\Psi}_{0}\rangle-\langle \frac{\hat\rho}{\hat D(\lam)},\widehat{{\cal P}r_0\rho}\rangle\Big),
\eeqn
Similarly,
\be\la{F1-split}
\langle \frac{\hat K_0}{\hat D(\lam)},J\hat\varrho\rangle=i\langle \frac{J\na\hat\rho}{\hat D(\lam)},\ov{\hat K}_0\rangle
=-|v|\langle \frac{k_1J\na\hat\rho}{\hat D(\lam)},\ov{\hat\Lam}_{0}\rangle
+i\lam\langle \frac{J\na\hat\rho}{\hat D(\lam)},\ov{\hat\Lam}_{0}\rangle
+i\langle  \frac{J\na\hat\rho}{\hat D(\lam)}, \ov{\hat\Psi}_{0}\rangle-i\langle \frac{J\na\hat\rho}{\hat D(\lam)},\widehat{{\cal P}r_0\rho}\rangle.
\ee
Therefore, asymptotics \eqref{int-rho-f-as}--\eqref{int-rho-f-as+}  together with \eqref{V-as}  imply 
$$
G_0^+(\lam)=g_0+g_{1}\lam+g_{2}\lam^2+g_3\lam^3+{\cal O}(\lam^4),\quad \lam\to 0,\quad \rRe\lam>0.
$$
It remains to note  that
\be\la{F00}
g_0=G_0^+(0)=0,\qquad g_1=(G_0^+)'(0)=0
\ee
by the symplectic orthogonality conditions $\Omega(X_0,\tau_j)=0$. We prove \eqref{F00} in   Appendix C.
\end{proof}
\begin{lemma}
For $k=0,1,2$, $\beta>4$  and sufficiently large $B>0$, the bounds hold
\be\la{K-as}
|(G^+_0)^{(k)}(\lam)|
\le C|\lam|^{1-k}(1+\Vert(\Lam_0,\Psi_0)\Vert_{{\cal F}_\beta}),\quad |\lam|\ge B,\quad \rRe\lam>0.
\ee
\end{lemma}
\begin{proof}
By \eqref{g-as},  
$$
|\langle {\cal R}^{(k)}(\lam)\Psi_0,\rho_1\rangle|\le C\Vert {\cal R}^{(k)}(\lam)\Psi_0\Vert_{H^{-1}_{-\beta}}\Vert \rho_1\Vert_{H^1_{\beta}}
\le C_1 |\lam|^{-2}\Vert\Psi_0\Vert_{L^{2}_{\beta}},\quad |\lam|\ge B,\quad \rRe\lam>0,
$$
where $\rho_1$ is one of the functions $\rho$, $\na\rho$,  $y\na\rho$.
The same estimate hold with  $\na\Lam_0$ instead of $\Psi_0$. Finally, \eqref{g-as} implies
 $$
|\langle {\cal R}^{(k)}(\lam)\Lam_0,\rho\rangle| \le C\Vert {\cal R}^{(k)}\na\Lam_0\Vert_{H^{-2}_{-\beta}}\Vert \rho_2\Vert_{H^2_{\beta}}
\le C_1 |\lam|^{-3}\Vert\na\Lam_0\Vert_{L^{2}_{\beta}}.
$$
Here $\rho_2=F^{-1}_{k\to x}\frac {\hat\rho(k)}{|k|}$.
Hence, \eqref{K-as}  follows by \eqref{F1+},  \eqref{F-split}, \eqref{F1-split} 
\end{proof}

\subsection{Spectral condition}\la{SpCon}
We suppose that  
\be\la{M-condition}
{\rm det}\, L(i\mu+0)=(a_{11}(i\mu+0)+b_{11}(i\mu+0))(a_{22}(i\mu+0)+b_{22}(i\mu+0)-a_{12}^2(i\mu+0)\ne 0.
\ee
for $\mu\ne 0$. Here
\beqn\nonumber
a_{jj}(i\mu+0)&=&-\mu^2 m+v^2\int\frac{k_j^2k_2^2|\hat\rho|^2}{k^2\hat D_0}dk+{\w}^2\int\frac{k_j^2|\na\hat\rho|^2}{\hat D_0}dk\\
\nonumber
&-&{\w}^2\int \frac{k_j^2|\na\hat\rho|^2}{\hat D(i\mu+0)}dk-\int\frac{(|v|k_1+(-1)^{j+1}\mu)^2k_2^2|\hat\rho|^2}{k^2\hat D(i\mu+0)}dk,\quad j=1,2,\\
\nonumber
a_{12}(i\mu+0)&=&i{\w}\mu\Big[\int\frac{(k\cdot\na\hat\rho)\hat\rho}{\hat D_0}-\int\frac{(k\cdot\na\hat\rho)\hat\rho}{\hat D(i\mu+0)}\Big]\\
\nonumber
b_{11}(i\mu+0)&=&\frac{1}{I+\varkappa(i\mu+0)}\Big[{\w}^2\Big(\int\frac{k_1|\na\hat\rho|^2}{\hat D(i\mu+0)} dk\Big)^2
+\mu^2\Big(\int\frac{\hat\rho \na_2\hat\rho}{\hat D(i\mu+0)} dk\Big)^2\Big]\\
\nonumber
b_{22}(i\mu+0)&=&\frac{v^2}{I+\varkappa(i\mu+0)}\Big(\int\frac{k_2\hat\rho \na_2\hat\rho}{\hat D(i\mu+0)} dk\Big)^2,
\eeqn
by \eqref{c-ro}, \eqref{h-ro}, \eqref{f-ro}, \eqref{Sj}, \eqref{tau-p}  and \eqref{aa}--\eqref{bb}.
In Appendix D we show that  condition \ref{M-condition} holds for $v\ne 0$ and sufficiently large $I$.
\subsection{Proof of proposition \ref{lindecay}}
{\bf Time decay of the vector components}\\
The component $r(t)$ is given by 
\be\la{F-int}
r(t)=\frac{1}{2\pi}\int e^{i\mu t}  L^{-1}(i\mu+0)G_0^+(\mu+i0) d\mu\\
\ee
Now \eqref{M-as}, \eqref{M-as-}, \eqref{F12-as},  and \eqref{K-as} imply  \eqref{frozenest}  for $r(t)$ by integration by parts two times.
\\
Due to \eqref{F1}, \eqref{F1+},  \eqref{aa} and  \eqref{detM} and \eqref{K-as},
\be\la{tir-as}
\ti r(\lam)=\frac{r_0}{\lam}+{\cal O}(|\lam|^{-2}),\quad r^{(k)}(\lam)={\cal O}(|\lam|^{-1-k}), \quad k=1,2, \quad \lam\to \infty, \quad \rRe\lam>0.
\ee
Hence,  the third equation of \eqref{eq-main} implies
$$
\ti\phi(\lam)=\frac{\om}m P\ti r+\lam\ti r-r_0=-\frac{\om}{m\lam} Pr_0+{\cal O}(|\lam|^{-2}),\quad \lam\to \infty, \quad \rRe\lam>0,
$$
Moreover,  $\hat\phi(i\mu+0)\in C^2(\R)$. Hence  double  integration by parts   gives
\be\la{Phi-decay}
\phi(t)=\frac{1}{2\pi}\int e^{i\mu t} \ti\phi(i\mu+0) d\mu={\cal O}(t^{-2}), \quad t\to \infty.
\ee
Similarly, definition \eqref{nu-rep} of $\ti\nu$, asymptotics of type \eqref{q-as} for $\varkappa(\lam)$, $f(\lam)$, $h(\lam)$  and $P_{jl}(\lam)$,  
asymptotics \eqref{int-rho-f-as}--\eqref{int-rho-f-as+}  and \eqref{tir-as} imply that $\ti\nu^{(k)}(\lam)={\cal O}(|\lam|^{-1-k})$. Hence,
\be\la{nu-decay}
\nu(t)=\frac{1}{2\pi}\int e^{i\mu t} \ti\nu(i\mu+0) d\mu={\cal O}(t^{-2}), \quad t\to \infty.
\ee
{\bf Time decay of the fields}
Denote ${\cal F}(t)=(\Lam(t),\Psi(t)$ and rewrite the first two equations of \eqref{bfA} as
\be\la{F-eq}
\dot{\cal F}(t)=\begin{bmatrix} v\cdot\na &1\\ 
\Delta& v\cdot\na\end{bmatrix}{\cal F}(t)+\begin{bmatrix} 0\\ Q(t)\end{bmatrix},
\ee
where
$$
Q={\cal P}[\rho\phi - \frac {\om}m\rho J Pr-\!v(r\cdot\na\rho)]+\big(\om(r\cdot \na)-\!\nu\big)J\varrho.
$$
Applying the Duhamel representation, we get
\be\la{Duhamel}
{\cal F}(t)=W_v(t){\cal F}_0+\int_0^t W_v(t-s)\begin{bmatrix} 0\\ Q(s)\end{bmatrix} ds,\quad t\ge 0,
\ee
where $W_v(t)$ is the dynamical group of the modified wave equation (equation \eqref{F-eq} with $Q(t)=0$).
For the group  $W_v(t)$ the dispersion decay holds \cite{K2010, KK2023}:
\be\la{dede}
\Vert W_v(t){\cal F}_0\Vert_{{\cal F}_{-\beta}}\le C(1+t)^{-2}\Vert{\cal F}_0\Vert_{{\cal F}_\beta},\quad \beta>2,\quad |v|<1,\quad t\ge 0.
\ee
 Applying  \eqref{dede} to \eqref{Duhamel}, and  using   \eqref {tir-as},  \eqref {Phi-decay}  and \eqref {nu-decay} we obtain \eqref{frozenest} 
 for the field components.
\setcounter{equation}{0}
\section{Appendix}
\subsection*{A. Proof of lemma \ref{Ome} }\label{ApA}
Here we compute the matrix elements $\Omega_{j,l}=\Omega_{j,l}(v):=\Omega(\tau_j(v),\tau_l(v))$ 
and prove that the matrix $\bf Q(v)$ is nondegenerate.
We will write  $A=A_{v}$, $\Pi=\Pi_v$, $p=p_{v}$.
By \eqref{solit3}
\beqn\la{APi}
\hat A(k)&=&\Big(\frac1{\hat D_0}v -\frac{(vk)}{k^2\hat D_0}k\Big)\hat \rho(k)+i\frac{\w}{\hat D_0}J\na\hat\rho(k)
=A^{+}(k)+iA^{-}(k)\\
\la{APi+}
\hat\Pi(k)&=&i(vk)(A^{+}(k)+iA^{-}(k)), \quad (vk)=v\cdot k.
\eeqn
where $\hat A^+$ is even, and $\hat A^-$ is  odd.  For $ j,l=1,2$, formulas \eqref{OmJ}--\eqref{inb} imply 
\be\label{jl}
\Om_{jl}=-2i\!\int\! k_jk_l(vk)\big(\hat A^{+}+i\hat A^{-}\big)\cdot\big(\hat A^{+}-i\hat A^{-}\big)dk
=-2i\!\int \!k_jk_l(vk)(|\hat A^{+}|^2+|\hat A^{-}|^2)dk=0,
\ee
Similarly, for $j=1,2$ one has
\beqn\nonumber
\Om_{j+2,j+2}&\!\!\!=\!\!\!&-\int \Big(\big(\pa_{v_j}(\hat A^{+}+i\hat A^{-})\big)
\cdot\big(i(vk)\pa_{v_j}(\hat A^{+}-i\hat A^{-})+ik_j(\hat A^{+}-i\hat A^{-})\big)\\
\nonumber
&\!\!\!+\!\!\!&\big( i(vk)\pa_{v_j}(\hat A^{+}+i\hat A^{-})+ik_j(\hat A^{+}+i\hat A^{-})\big)\cdot\big(\pa_{v_j}(\hat A^{+}-i\hat A^{-})\big)\Big)dk\\
\la{33}
&\!\!\!=\!\!\!&-2i\!\int\!\Big( (vk)\big(|\pa_{v_j}\hat A^{+}|^2+|\pa_{v_j}\hat A^{-}|^2\big)
+k_j(\hat A^{+}\!\cdot \pa_{v_j}\hat A^{+}\!+\hat A^{-}\!\cdot \pa_{v_j}\hat A^{-})\Big) dk=0.
\eeqn
By \eqref{om-M} and  \eqref{APi},
\beqn\la{paA2}
&&\pa_{v_l}\hat A^{+}=\frac{1}{\hat D_0}\Big(e_l+\frac{2(vk)k_l}{\hat D_0}v-\frac{k_l(k^2+(vk)^2)}{k^2\hat D_0}k\Big)\hat\rho,\quad
\\
\la{paA1}
&&\pa_{v_l}\hat A^{-}=\Big(\frac{2(vk)k_l}{\hat D_0}-\frac{1}{(I+\varkappa)}\int \frac{2(vk)k_l|\na\hat\rho|^2}{\hat D^2_0}dk\Big)\frac{{\w} J\na\hat\rho}{\hat D_0},
\eeqn
 Moreover,
\be\la{paP}
\pa_{v_l}p= me_l+\int\frac{|\hat\rho|^2dk}{\hat D_0}e_l+\int\frac{2|\hat\rho|^2(vk)k_ldk}{\hat D_0^2}v
-\int\frac{|\hat\rho|^2(k^2+(vk)^2)k_ldk}{k^2\hat D_0^2}k
\ee
 by \eqref{solY} and \eqref{paA2}. Formulas \eqref{OmJ}, \eqref{inb} and \eqref{paP} imply  that
\beqn\nonumber
&&\Om_{j,l+2}=\langle ik_j\hat A, i(vk)\pa_{v_l}\hat A+ik_l\hat A \rangle
+\langle k_j(kv)\hat A, \pa_{v_l}\hat A\rangle+e_j\cdot \pa_{v_l}p\\
\nonumber
&&=\int k_jk_l(|\hat A^{+}|^2+|\hat A^{-}|^2)dk+2\int (vk)k_j(\hat A^{+}\cdot\pa_{v_l}\hat A^{+}+\hat A^{-}\cdot\pa_{v_l}\hat A^{-})\\
\la{Om14}
&&+\big(m+\!\int\frac{|\hat\rho|^2dk}{\hat D_0}\big)\delta_{jl}+\!\int\frac{2|\hat\rho|^2(vk)v_jk_ldk}{\hat D_0^2}
-\!\int\frac{|\hat\rho|^2(k^2+(vk)^2)k_jk_ldk}{k^2\hat D_0^2},~~j,l=1,2.
\eeqn
Now we  choose  $v=(|v|,0)$. In this case $(vk)=|v|k_1$, and  $\hat D_0:=k^2-v^2k_1^2$. Hence,
\be\la{14}
\Om_{14}=\Om_{23}=0.
\ee 
Finally, \eqref{jl}, \eqref{33} and \eqref{14}  imply that
\be\la{Om-fin}
{\rm det}\,\Om(v,w)={\rm det}\begin{bmatrix}
0&0& \Om_{13} &0 \\
0&0& 0 & \Om_{24} \\
\Om_{31} & 0 &0&\Om_{34}\\
0 & \Om_{42} &\Om_{43}& 0\\
\end{bmatrix}
=\Om_{13}^2 \Om_{24}^2 
\ee
since $\Om_{jl}=-\Om_{lj}$ for $j\ne l$. Let us   prove that 
$\Om_{13} \Om_{24} \ne 0$.
In the case   $v=(|v|,0)$, \eqref{APi} and  \eqref{paA2}--\eqref{paA1} become
\beqn\la{solit32+} 
\hat A^{+}&\!\!\!=\!\!\!&\frac{|v|k_2}{k^2\hat D_0}\hat\rho Jk,\quad  \hat A^{-}=\frac{\w}{\hat D_0} J\na\hat\rho, \quad
|\hat A^{+}|^2=\frac{v^2k_2^2}{k^2\hat D_0^2}|\hat\rho|^2,\quad |\hat A^{-}|^2=\frac{\w^2}{\hat D_0^2}|\na\hat\rho|^2,\\ 
\la{A1A1}
\pa_{v_1}\hat A^{+}&\!\!\!=\!\!\!&\frac{1}{\hat D_0}\Big(e_1+\frac{2v^2k_1^2}{\hat D_0}e_1-\frac{k_1(k^2+v^2k_1^2)}{k^2\hat D_0}k\Big)\hat\rho
=\frac{k_2(k^2+v^2k_1^2)}{k^2\hat D_0^2}\hat\rho Jk,\\
\la{A1A1+}
\pa_{v_2}\hat A^{+}&\!\!\!=\!\!\!&\frac{1}{\hat D_0}\Big(e_2+\frac{2v^2k_1k_2}{\hat D_0}e_1-\frac{k_2(k^2+v^2k_1^2)}{k^2\hat D_0}k\Big)\hat\rho
=\frac{k_1(k^2(v^2-1)+v^2k_2^2)}{k^2\hat D_0^2}\hat\rho Jk,\\
\la{A1A1++}
\pa_{v_1}\hat A^{-}&\!\!\!=\!\!\!&
\Big(\frac{2|v|k_1^2}{\hat D_0}-\frac{1}{(I+\varkappa)}\int \frac{2|v|k_1^2|\na\hat\rho|^2}{\hat D^2_0}dk\Big)\frac{{\w} J\na\hat\rho}{\hat D_0},\qquad 
\pa_{v_2}\hat A^{-}=\frac{2{\w} |v|k_1k_2}{\hat D_0^2}J\na\hat\rho.
\eeqn
Hence,
\beqn\nonumber
\int |v|k_1^2(\hat A^{+}\cdot\pa_{v_1}\hat A^{+})dk&=&v^2\int \frac{k_1^2k_2^2(k^2+v^2k_1^2)}{k^2\hat D_0^3}|\hat\rho|^2\ge 0,\\
\la{AdA+} 
\int 2v|k_1^2(\hat A^{-}\cdot \pa_{v_1}\hat A^{-})dk&=&2{\w}^2 v^2\Big[\int\frac{k_1^4 |\nabla\hat\rho|^2dk}{\hat D_0^3}
-\frac{1}{(I+\varkappa)}\Big(\int\frac{k_1^2 |\nabla\hat\rho|^2dk}{\hat D_0^2}\Big)^2\Big]\ge 0
\eeqn
since
$$
\Big(\int\frac{k_1^2 |\nabla\hat\rho|^2dk}{\hat D_0^2}\Big)^2\le\int\frac{|\nabla\hat\rho|^2dk}{\hat D_0} \int\frac{k_1^4 |\nabla\hat\rho|^2dk}{\hat D_0^3}
=\varkappa \int\frac{k_1^4 |\nabla\hat\rho|^2dk}{\hat D_0^3}.
$$
Let us also note that
\be\la{ep11}
\int\frac{|\hat\rho|^2dk}{\hat D_0}+\int\frac{2v^2k_1^2|\hat\rho|^2dk}{\hat D_0^2}
-\int\frac{|\hat\rho|^2(k^2+v^2k_1^2)k_1^2dk}{k^2\hat D_0^2}=\int\frac{k_2^2(k^2+v^2k_1^2)}{k^2\hat D_0^2}|\hat\rho|^2dk\ge 0
\ee
Therefore, \eqref{Om14} and \eqref{AdA+}--\eqref{ep11}  imply that $\Om_{13}>0$. 
It remains to prove that $\Om_{24}>0$. By \eqref{solit32+} and \eqref{A1A1++},
$$
\hat A^{+}\!\cdot \pa_{v_2}\hat A^{+}=\frac{|v|k_1k_2}{\hat D_0^3}\Big(2v^2-1-\frac{v^2k_1^2}{k^2}\Big)|\hat\rho|^2,\qquad
\hat A^{-}\!\cdot \pa_{v_2}\hat A^{-}\!=\frac{2{\w}^2|v| k_1k_2}{\hat D_0^3} |\nabla\hat\rho|^2.
$$
Hence, \eqref{Om14} and \eqref{solit32+}  imply that
$$
\Om_{24}=m+\int X_{24}|\hat\rho|^2dk+\int \Big(\frac{\om^2k_2^2}{\hat D_0^2}+\frac{4\om^2 v^2k_2^2k_1^2}{\hat D_0^3}\Big)|\na\hat\rho|^2dk>0
$$
since
\beqn\nonumber
X_{24}&=&\frac{k_2^2}{\hat D_0^2}\Big(v^2-\frac{v^2k_1^2}{k^2}+\frac{4v^4k_1^2}{\hat D_0}-\frac{2v^2k_1^2}{\hat D_0}
 -\frac{2v^4k_1^4}{k^2\hat D_0}-1- \frac{v^2k_1^2}{k^2}\Big)+\frac{1}{\hat D_0}\\ 
 \nonumber
  &=&\frac{k_2^2}{\hat D_0^3}\Big[v^2k^2+3v^4k_1^2-3v^2k_1^2-k^2\Big]+\frac{k^4+v^4k_1^4-2v^2k_1^2k^2}{\hat D_0^3}\\
\nonumber
&=&\frac{k_2^2}{\hat D_0^3}\Big[v^2(k^2-k_1^2)+2v^4k_1^2+(v^2k_1-k_1)^2-k_1^2\Big]+\frac{1}{\hat D_0^3}\Big[(k^2(k^2-k_2^2)+v^4k_1^4-2v^2k_1^2k^2\Big]\\
\nonumber
&\ge&\frac{1}{\hat D_0^3}\Big[v^4k_2^4+2v^4k_1^2k_2^2+v^4k_1^4+k_1^4-2v^2k_1^2k^2\Big]
=\frac{(v^2k^2-k_1^2)^2}{\hat D_0^3}\ge 0.
\eeqn
\subsection*{B. Proof of Lemma \ref{ljf}}\label{ApB}
We will again write  $A=A_{v}$, $\Pi=\Pi_v$, $p=p_{v}$.  Differentiating  (\ref{Hs5}) in $b_j$ and  $v_j$, we obtain 
\beqn\la{d1}
\!\!\!\!\!\!\!\!-(v\cdot\na)\pa_j A&=&\pa_j\Pi,\\
\la{d11}
\!\!\!\!\!\!\!\!-(v\cdot\na)\pa_j\Pi&=&\Delta\pa_j A+{\cal P}[v\pa_j\rho] -{\w}\pa_j J\varrho,\\
\la{d2}
\!\!\!\!\!\!\!\!-\pa_j A-(v\cdot\na)\pa_{v_j} A&=&\pa_{v_j}\Pi,\\
\la{d21}
\!\!\!\!\!\!\!\!-\pa_j\Pi- (v\cdot\na)\pa_{v_j}\Pi&=&\De\pa_{v_j} A-\frac{1}{I}\langle\pa_{v_j} A,J\varrho\rangle J\varrho+{\cal P}[e_j\rho]
\eeqn
\\
{\it Step i)} 
By \eqref{AA}, the first equation of \eqref{Atan} reads
\beqn\la{e1}
&&\!\!\!\!\!\!\!\!\!\!\!\!\!\!\!\!\!\!\!\!\!\!\!\!\!(u\cdot\na)\pa_jA+\pa_j \Pi=((u-v)\cdot\nabla) \pa_j A\\
\nonumber
&&\!\!\!\!\!\!\!\!\!\!\!\!\!\!\!\!\!\!\!\!\!\!\!\!\!\Delta \pa_j A\!-\!{\cal P}\big[\frac{\langle\pa_j A,\rho\rangle\rho}{m}-\frac{{\w}\rho}{m}JPe_j\!-\!v\pa_j\rho\big]\!
-\!\Big(\frac{\langle\pa_j A,J\varrho\rangle}{I}+\!\frac{e_j\!\cdot \!PJv}{I}+{\w}\pa_j\Big)J\varrho+(u\cdot\na)\pa_j\Pi\\
\la{e2}
&&=((u-v)\cdot\!\na)\pa_j\Pi\\
\la{e3}
&&\!\!\!\!\!\!\!\!\!\!\!\!\!\!\!\!\!\!\!\!\!\!\!\!\!\langle \pa_j A,\rho\rangle-{\w} JPe_j= 0\\
\la{e4}
&&\!\!\!\!\!\!\!\!\!\!\!\!\!\!\!\!\!\!\!\!\!\!\!\!\!\langle v\cdot\pa_j A,\na\rho\rangle-\frac {\w}m\langle  PJ\pa_j A,\rho\rangle
-{\w}\langle \na\cdot J\pa_j A,\varrho\rangle -\frac{1}{I}PJv\langle\pa_j A,J\varrho\rangle -Se_j=0
\eeqn
Equation \eqref{e1}  follows from \eqref{d1}.
Further,  \eqref{solit3} and  \eqref{c-ro} imply \eqref{e3} since
\be\la{d4}
\langle \pa_j A,\rho\rangle={\w}\langle \frac{k_jJ\na\hat\rho}{\hat D_0},\hat\rho\rangle = {\w} JPe_j.
\ee
By \eqref{solit3},
\be\la{e5}
\langle \pa_j A,J\varrho\rangle=\langle \frac{k_j\hat\rho(v-\frac{(vk)}{k^2}k)}{\hat D_0},J\na\hat\rho\rangle=
\langle \frac{k_j\hat\rho}{\hat D_0}v,J\na\hat\rho\rangle=-e_j\cdot PJv.
\ee
since $k\cdot J\na\rho=0$. Now \eqref{e2} follows from \eqref{d11}, \eqref{d4} and \eqref{e5}.
\smallskip\\
It remains to prove  \eqref{e4}. By \eqref{solit3}, \eqref{c-ro}, \eqref{h-ro} and  \eqref{f-ro} ,
\beqn\nonumber
 \langle PJ\pa_j A,\rho\rangle&\!\!\!=\!\!\!&-{\w} P\langle \frac{k_j}{\hat D_0}\na\hat\rho,\hat\rho\rangle=-{\w} P^2e_j,\\  
\nonumber
\langle v\cdot\pa_j A,\na\rho\rangle&\!\!\!=\!\!\!&\int\frac{k_jv^2|\hat\rho|^2}{\hat D_0}k dk-\int\frac{k_j(vk)^2|\hat\rho|^2}{k^2\hat D_0} kdk= Qe_j,\\
\nonumber
\langle \na\cdot J\pa_j A,\varrho\rangle
&\!\!\!=\!\!\!&-{\w}\langle \frac{k_j(k_1\na_1\hat\rho+k_2\na_2\hat\rho)}{\hat D_0},\na\hat\rho\rangle
=-{\w}\int\frac{k^2}{\hat D_0}\begin{pmatrix}\na_j\hat\rho\na_1\hat\rho\\ \na_j\hat\rho\na_2\hat\rho\end{pmatrix} =-{\w} Fe_j.
\eeqn
Hence, \eqref{e4} follows by \eqref{j-al}. 
 \smallskip\\
{\it Step ii)} 
By \eqref{AA},  the second equation of\eqref{Atan}  reads  
\beqn
\la{e21}
&&\!\!\!\!\!\!\!\!\!\!\!\!\!\!\!\!\!(u\cdot\na) \pa_{v_j} A+\pa_{v_j}\Pi =((u-v)\cdot\na) \pa_{v_j} A-\pa_j A\\
\la{e22}
&&\!\!\!\!\!\!\!\!\!\!\!\!\!\!\!\!\!\!\!\De\pa_{v_j} A-\!{\cal P}\big[\frac{\rho}{m}\big(\langle\pa_{v_j} A,\rho\rangle\!-\!\pa_{v_j}p\big)\big]\!
-\!\frac{\langle \pa_{v_j }A,J\varrho\rangle}{I} J\varrho+(u\cdot\!\na) \pa_{v_j} \Pi=((u\!-\!v)\!\cdot\!\na) \pa_{v_j} \Pi-\pa_j \Pi\\
\la{e23}
&&\!\!\!\!\!\!\!\!\!\!\!\!\!\!\!\!\!-\langle \pa_{v_j} A,\rho\rangle+\pa_{v_j}p=me_j\\
\la{e24}
&&\!\!\!\!\!\!\!\!\!\!\!\!\!\!\!\!\!\frac {\w}m PJ\big(\langle \pa_{v_j}A,\rho\rangle\!-\!\pa_{v_j}p\big)\!-\!\langle v\cdot\pa_{v_j} A,\!\na\rho\rangle\! 
\!+\!{\w}\langle\na\!\cdot\! J\pa_{v_j}A,\varrho\rangle\!+\!\frac{1}{I}PJv\langle\pa_{v_j}A,J\varrho\rangle\!=\!0
\eeqn
Equation \eqref{e21} and \eqref{e23} immediately follows from \eqref{d2} and  \eqref{solY}, respectively. 
Equation \eqref{e22} follows from \eqref{d21} and \eqref{e23}.  
It remains to prove  \eqref{e24}. By \eqref{om-M}, 
$$
\pa_{v_j}\langle A,J\varrho\rangle=I\pa_{v_j}\w. 
$$
 Hence, using  \eqref{e23}, we rewrite \eqref{e24} as
\be\la{5}
-{\w} PJe_j-v\cdot\pa_{v_j} \langle A,\na\rho\rangle+{\w} \pa_{v_j}\langle\na\cdot JA,\varrho\rangle+PJv \pa_{v_j}\w=0
\ee
By \eqref{JAro},
\be\la{51}
\pa_{v_j}\langle\na\cdot JA,\varrho\rangle=\pa_{v_j} (PJv)=PJe_j+(\pa_{v_j} P)Jv
\ee
Further, \eqref{c-ro} implies
$$
\langle A_1,\na\rho\rangle=-{\w}\langle \frac{\na_2\hat\rho}{\hat D_0},k\hat\rho\rangle=-{\w}\begin{pmatrix}P_{21}\\P_{22}\end{pmatrix},
\quad
\langle A_2,\na\rho\rangle={\w}\langle \frac{\na_1\hat\rho}{\hat D_0},k\hat\rho\rangle={\w}\begin{pmatrix}P_{11}\\P_{12}\end{pmatrix}
$$
Hence,
\be\la{52}
v\cdot\pa_{v_j} \langle A,\na\rho\rangle={\w} (\pa_{v_j} P)Jv+PJv \pa_{v_j}\w.
\ee
Finally,  \eqref{51}--\eqref{52} imply \eqref{5}.
\subsection*{C. Symplectic orthogonality conditions}
Here we prove \eqref{F00}. By \eqref{inb} and \eqref{APi}, 
\beqn\nonumber
0&=&\Omega(X_0,\tau_j)=-\langle\Lam_0,\pa_j \Pi\rangle+\langle\Psi_0,\pa_j A\rangle-\pi_0\cdot e_j\\
\nonumber
&=&-\int (\hat\Lam_0\cdot v) \frac{k_j(vk)\rho}{\hat D_0}dk+i\om\int (\hat\Lam_0\cdot J\na\hat\rho) \frac{k_j(vk)}{\hat D_0} dk\\
\la{soc0}
&+&i\int (\hat\Psi_0\cdot v)\frac{k_j\rho}{\hat D_0} dk+\om\int (\hat\Psi_0\cdot  J\na\hat\rho)\frac{k_j}{\hat D_0} dk -\pi_0\cdot e_j,\quad j=1,2
\eeqn
since $\Lam_0\cdot k=0$ and $\Psi_0\cdot k=0$ due to the solenoidality of the  fields $\Lam_0$ and $\Psi_0$.
For $v=(|v|,0)$, in the vector form, \eqref{soc0} reads
\be\la{soc1}
v^2\int\frac{ \hat\Lam_{01}k_1\rho}{\hat D_0} kdk-i\om|v|\int (\hat\Lam_0\cdot J\na\hat\rho)\frac{k_1}{\hat D_0} kdk
-i|v|\int\frac{ \hat\Psi_{01}\rho}{\hat D_0} k dk-\om\int\frac{\hat\Psi_0\cdot  J\na\hat\rho}{\hat D_0} kdk+\pi_0=0
\ee
Further, \eqref{solY} and \eqref{inb} imply
\be\la{X0-2}
0=\Omega(X_0,\tau_{j+2})=\langle\Lam_0, \pa_{v_j} \Pi\rangle-\langle\Psi_0,\pa_{v_j} A\rangle+r\cdot (me_j+\langle\pa_{v_j} A,\rho\rangle).
\ee
By \eqref{A1A1}--\eqref{A1A1+}, 
\beqn\nonumber
\langle \pa_{v_1}A,\rho\rangle&=&\langle\frac{k_2(k^2+v^2k_1^2)}{k^2\hat D_0^2}\hat\rho Jk,\hat\rho\rangle
=e_1\int\frac{(v^2k_2^2k_1^2+k^2k_2^2)|\hat\rho|^2}{k^2\hat D_0^2}dk\\ 
\nonumber
\langle \pa_{v_2}A,\rho\rangle&=&\langle\frac{k_1(k^2(v^2-1)+v^2k_2^2)}{k^2\hat D_0^2}\hat\rho Jk,\hat\rho\rangle
=-e_2\int\frac{(v^2k_2^2k_1^2+k_2k_1^2(v^2-1))|\hat\rho|^2}{k^2\hat D_0^2}dk
\eeqn
Moreover, \eqref{A1A1}--\eqref{A1A1++} imply 
\beqn\nonumber
\langle\Psi_0,\pa_{v_1} A\rangle&=&\ds\int\Big(\frac{\hat\Psi_{01}}{\hat D_0}+2v^2\frac{ \hat\Psi_{01}k_1^2}{\hat D_0^2}\Big)\hat\rho dk
-2i|v|\om\int \frac{k_1^2(\hat\Psi_0\cdot J\na\hat\rho)}{\hat D_0^2}kdk\\
\nonumber
&+&\frac{2i|v|\om}{1+\varkappa}\int\frac{k_1^2|\na\hat\rho|^2 dk}{\hat D_0^2}\int\frac{\hat\Psi_0\cdot J\na\hat\rho}{\hat D_0}dk, \\
\nonumber
\langle\Psi_0,\pa_{v_2} A\rangle&=&\ds\int\Big(\frac{\hat\Psi_{02}}{\hat D_0}+2v^2\frac{ \hat\Psi_{01}k_1k_2}{\hat D_0^2}\Big)\hat\rho dk
-2i|v|\om\int \frac{k_1k_2(\hat\Psi_0\cdot J\na\hat\rho)}{\hat D_0^2}kdk
\eeqn
In the vector form, the last two equations   read
\beqn\nonumber
\langle\Psi_0,\pa_{v_j} A\rangle&=&\ds\int\Big(\frac{\hat\Psi_{0j}}{\hat D_0}+2v^2\frac{ \hat\Psi_{01j}k_1k_j}{\hat D_0^2}\Big)\hat\rho dk
-2i|v|\om\int \frac{k_1k_j(\hat\Psi_0\cdot J\na\hat\rho)}{\hat D_0^2}kdk\\
\la{vec-f}
&+&\delta_{1j}\frac{2i|v|\om}{I+\varkappa}\int\frac{k_1^2|\na\hat\rho|^2 dk}{\hat D_0^2}\int\frac{\hat\Psi_0\cdot J\na\hat\rho}{\hat D_0}dk 
\eeqn
By \eqref{APi+} and \eqref{solit32+}- \eqref{A1A1++}, 
\beqn\nonumber
\pa_{v_j}\Pi& =& ik_jA+i|v|k_1\pa_{v_j} A=ik_j(\frac{(|v|e_1-\frac{|v|k_1k}{k^2})\hat \rho(k)}{\hat D_0}+\frac{i\omega J\na\hat\rho(k)}{\hat D_0})\\
\nonumber
&+&i|v|k_1\frac{1}{\hat D_0}\Big(e_j+\frac{2v^2k_1k_j}{\hat D_0}e_1-\frac{k_1(k_1k_j+v^2k_1^2)}{k^2\hat D_0}k\Big)\hat\rho\\
\nonumber
&-&2v^2\omega k_1\Big(\frac{k_1k_j}{\hat D_0}+\frac{\delta_{1j}}{(I+\varkappa)}\int \frac{k_1^2|\na\hat\rho|^2}{\hat D^2_0}dk\Big)\frac{J\na\hat\rho}{\hat D_0}
\eeqn
Therefore
\beqn\nonumber
\langle\Lam_0, \pa_{v_j} \Pi\rangle&=&-i|v|\int\Big(\frac{\hat\Lam_{01}k_j}{\hat D_0}
+\frac{\hat\Lam_{0j}k_1}{\hat D_0}+2v^2\frac{\hat\Lam_{01}k_1^2k_j}{\hat D_0^2}\Big)\hat\rho dk-\om\int \frac{k_j(\hat\Lam_0\cdot J\na\hat\rho)}{\hat D_0}dk\\
\la{vec-f1}
&-&2v^2\om\int \frac{k_1^2k_j(\hat\Lam_0\cdot J\na\hat\rho)}{\hat D_0^2}kdk
-\frac{2v^2\omega\delta_{1j}}{(I+\varkappa)}\int \frac{k_1^2|\na\hat\rho|^2}{\hat D^2_0}dk\int\frac{k_1\hat\Lam_0\cdot J\na\hat\rho}{\hat D_0}
\eeqn
Hence, the symplectic orthogonality condition \eqref{X0-2} together with \eqref{vec-f}-\eqref{vec-f1} give
\beqn\nonumber
\!\!\!\!\!0&\!\!=\!\!&-i|v|\int \frac{\hat\Lam_{01}(k^2+v^2k_1^2)\rho}{\hat D_0^2}k dk-i|v|\int \frac{k_1\rho}{\hat D_0}\Lam_{0} dk
-\om\int (\hat\Lam_0\cdot J\na\hat\rho) \frac{k^2+v^2k_1^2}{\hat D_0^2}kdk\\
\nonumber
&\!\!-\!\!&2v^2\!\int \frac{ \hat\Psi_{01}k_1\rho}{\hat D_0^2}kdk-\int \frac{\rho}{\hat D_0}\Psi_{0} dk
+2i|v|\om\!\int (\hat\Psi_0\cdot J\na\hat\rho) \frac{k_1}{\hat D_0^2}kdk\\
\nonumber
&\!\!+\!\!&m r_{0}+r_{01}e_1\int\frac{(v^2k_2^2k_1^2+k^2k_2^2)|\hat\rho|^2}{k^2\hat D_0^2}dk
-r_{02}e_2\int\frac{(v^2k_2^2k_1^2+k^2k_1^2(v^2-1))|\hat\rho|^2}{k^2\hat D_0^2}dk\\
\la{soc2}
&\!\!-\!\!&\frac{2|v|\omega\delta_{1j}}{(I+\varkappa)}\int \frac{k_1^2|\na\hat\rho|^2}{\hat D^2_0}dk\int\frac{(i|v|k_1\hat\Lam_0+\hat\Psi_0)\cdot J\na\hat\rho}{\hat D_0}dk.
\eeqn
By \eqref{K0},
$$
\frac{K_0(0)}{\hat D_0}=\frac{i|v|k_1\hat\Lam_0+\hat\Psi_0-\widehat{{\cal P}[r_0\rho]}}{\hat D_0},\quad {\rm where}\quad
\widehat{{\cal P}[r_0\rho]}=\hat\rho\frac{r_0\cdot Jk}{k^2}Jk.
$$
Hence  \eqref{F1}, \eqref{F1+}  and \eqref{soc1} imply
\beqn\nonumber
&&G^+_0(0)=G_0(0)=\pi_0+\om PJr_0  -i|v|\langle \frac{\hat K_{01}(0)}{\hat D_0},k\hat\rho\rangle+\om\langle \frac{k\cdot J\hat K_0(0)}{\hat D_0},\na\hat\rho\rangle\\
\nonumber
&&=\pi_0+\om PJr_0+v^2\int\frac{k_1\hat\rho\hat\Lam_{01}}{\hat D_0} k dk-i|v|\int\frac{\hat\rho\hat\Psi_{01}}{\hat D_0} k dk\\
\la{F1-zero}
&&+i|v|\om\int\frac {k_1(k\cdot J\hat \Lam_0)}{\hat D_0} \na\hat\rho dk
+\om\int\frac {k\cdot J\hat \Psi_0}{\hat D_0} \na\hat\rho dk+\om\int \frac{\hat\rho (r_0\cdot Jk)}{\hat D_0}\na\hat\rho dk=0
\eeqn
since $PJr_0+\int \frac{\hat\rho (r_0\cdot Jk)}{\hat D_0}\na\hat\rho dk=0$, and
$$
\int\frac {k_1(k\cdot J\hat \Lam_0)}{\hat D_0} \na\hat\rho dk=-\int (\hat\Lam_0\cdot J\na\hat\rho)\frac{k_1}{\hat D_0} kdk,\quad
\int\frac {k\cdot J\hat \Psi_0}{\hat D_0} \na\hat\rho dk=-\int\frac{\hat\Psi_0\cdot  J\na\hat\rho}{\hat D_0}k dk.
$$
It remains to prove that $\pa_{\lam}G^+_0(0)=0$. By \eqref{K0}, 
\beqn\nonumber
\Big(\frac{K_0(\lam)}{\hat D(\lam)}\Big)'(0)&=&-2i|v|k_1\frac{i|v|k_1\hat\Lam_0+\hat\Psi_0
-\widehat{{\cal P}[r_0\rho]}}{\hat D^2_0}+\frac{\hat\Lam_0}{\hat D_0}\\
\nonumber
&=&\frac{(k^2+v^2k_1^2)\Lam_{0}-2i|v|k_1\hat\Psi_0} {\hat D^2_0}+\frac{2i|v|k_1\hat\rho(r_0\cdot Jk)Jk} {k^2\hat D^2_0}.
\eeqn
Hence,  \eqref{F1},  \eqref{fj-def}, \eqref{F1+} and \eqref{soc2}  imply
\beqn\nonumber
(G^+_0)'(0)\!\!&\!\!=\!\!&\!\!-i|v|\langle \big(\frac{K_{01}}{\hat D}\big)'(0),k\hat\rho\rangle
+\om\langle k\cdot J \big(\frac{K_{01}}{\hat D}\big)'(0),\na\hat\rho\rangle-\langle \hat K_0(0),\hat\rho\rangle\\
\nonumber
\!\!&\!\!+\!\!&\!\!mr_0-i\om\frac{\langle K_0(0),J\hat\varrho\rangle}{\kappa(0)}g'(0)\\
\nonumber
\!\!&\!\!=\!\!&\!\!-i|v|\Big(\langle\frac{(k^2\!+\!v^2k_1^2)\Lam_{01}} {\hat D^2_0},k\hat\rho\rangle
-2i|v|\langle\frac{k_1\Psi_{01}} {\hat D^2_0},k\hat\rho\rangle
+2i|v|\langle \frac{k_1k_2\hat\rho(r_0\cdot Jk)} {k^2\hat D^2_0},k\hat\rho\rangle\rangle\Big)\!+mr_0\\
\nonumber
\!\!&\!\!+\!\!&\!\!\om\Big(\langle \frac{(k^2+v^2k_1^2)(k\cdot J\Lam_0)}{\hat D^2_0},\na\hat\rho\rangle
-2i|v|\langle \frac{k_1(k\cdot J\Psi_0)}{\hat D^2_0},\na\hat\rho\rangle-2i|v|\langle \frac{k_1\hat\rho(r_0\cdot Jk)}{\hat D^2_0},\na\hat\rho\rangle\Big)\\
\nonumber
\!\!&\!\!-\!\!&\!\!i|v|\langle \frac{k_1\Lam_0}{\hat D_0},\hat\rho\rangle\! -\!\langle \frac{\Psi_0}{\hat D_0},\hat\rho\rangle
\!+\!\langle\frac{\hat\rho(r_0\cdot Jk)}{k^2\hat D_0}Jk,\hat\rho\rangle
-2\delta_{j1}\om|v|\frac{\langle K_0(0),J\hat\varrho\rangle}{1+\varkappa}\!\!\int\!\! \frac{k_1^2|\na_1\rho|^2}{\hat D_0^2}dk=0
\eeqn
since
\beqn\nonumber
&&2v^2\langle\frac{\hat\rho k_1k_2(r_0\cdot Jk)}{k^2\hat D_0^2},k\hat\rho\rangle+\langle\frac{\hat\rho(r_0\cdot Jk)}{k^2\hat D_0}Jk,\hat\rho\rangle\\
\nonumber
&&=r_{01}e_1\int\frac{(2v^2k_1^2k_2^2+k_2^2(k^2-v_1^2k_1^2)|\hat\rho|^2}{k^2\hat D_0^2}dk
-r_{02}e_2\int\frac{(2v^2k_1^2k_2^2-k_1^2(k^2-v_1^2k_1^2)|\hat\rho|^2}{k^2\hat D_0^2}dk\\
\nonumber
&&=r_{01}e_1\int\frac{(v^2k_1^2k_2^2+k_2^2k^2)|\hat\rho|^2}{k^2\hat D_0^2}
-r_{02}e_2\int\frac{(v^2k_2^2k_1^2+k^2k_1^2(v^2-1))|\hat\rho|^2}{k^2\hat D_0^2}dk
\eeqn
\subsection*{D. Regularity in continuous spectrum }
Let us assume that  $I=\infty$.  Im thia case   $\w=0$ and $b_{jl}=0$. Hence, condition \eqref{M-condition} becomes 
\beqn\nonumber
{\rm det}\, L(i\mu+0)&=&\Big[\mu^2 m-v^2\!\int\frac{k_1^2k_2^2|\hat\rho|^2}{k^2\hat D_0}dk+\int\!\frac{(|v|k_1+\mu)^2k_2^2|\hat\rho|^2}{k^2\hat D(i\mu+0)}dk\Big]\times\\
\la{M-condition+}
&\times&\Big[\mu^2 m-v^2\!\int\frac{k_2^4|\hat\rho|^2}{k^2\hat D_0}dk+\!\int\frac{(|v|k_1-\mu)^2k_2^2|\hat\rho|^2}{k^2\hat D(i\mu+0)}dk\Big]\ne 0.
\eeqn
\begin{lemma}\la{ek1}
Let  $0<|v|<1$. Then for any $\mu\ne 0$,
$$
\rIm \Big(\int\frac{h_{\pm}(k)\,dk}{\hat D(i\mu+0)}\Big)<0, \quad {\rm where} ~~h_{\pm}(k)=\frac{(|v|k_1\!-\!\mu)^2k_2^2|\hat\rho|^2}{k^2}.
$$
\end{lemma}
\begin{proof}
In the case  $v=(|v|,0)$ one has $\hat D(i\mu+0)=\hat D(v,k,i\mu+0)=k^2+(i|v|k_1+(i\mu+0))^2$. Denote
$K_v(\mu)=\{k\in\R^2: (k_1-\mu |v|\gamma^2)^2+k_2^2\gamma^2=\mu^2\gamma^4\}$, $\gamma=1/\sqrt{1-v^2}$. By the Plemelj formula,
$$
\rIm \Big(\int\frac{h_{\pm}(k)\,dk}{\hat D(i\mu+0)}\Big)=-\pi\int_{K_v(\mu)} \frac{h_{\pm}(k)}{|\na_k\hat D(i\mu+0)|}dS,
$$
where $h(k)\ge 0$ on $K_v(\mu)$. Hence, it suffices to show that  $h(k) \not{\!\!\equiv}\, 0$ on $K_v(\mu)$ for $\mu\ne 0$.
Assume that, on the contrary, $h_{\pm}(k) \equiv 0$ on $K_v(\mu)$.
Then  $\hat\rho(k) \equiv 0$ in the ring  $ |\mu|(1-|v|)\gamma^2 \le |k|\le |\mu|(1+|v|)\gamma^2$ by the symmetry (\ref{rosym}) of $\hat\rho(k)$. Hence,
$\hat\rho(k) \equiv 0$ in $\C^2$ due to the analyticity of $\hat\rho(k)$,
in contradiction to the last assumption in (\ref{rosym}).
 \end{proof}
 The lemma  implies that  
$$
\rIm \Big[\mu^2 m-v^2\!\int\frac{k_1^2k_2^2|\hat\rho|^2}{k^2\hat D_0}dk+\int\!\frac{(|v|k_1\pm \mu)^2k_2^2|\hat\rho|^2}{k^2\hat D(i\mu+0)}dk\Big]<0,\quad  \mu\ne 0.
$$ 
 Thus, in the case  $v\ne0$ and $I=\infty$, ${\rm det}\, L(i\mu+0)\ne 0$   by \eqref{M-condition+}.  
Therefore ${\rm det}\, L(i\mu+0)\ne 0$ for large $I$ by the continuity of ${\rm det}\, L(i\mu+0)$.

\end{document}